\documentclass[final]{article}%
\usepackage{graphicx}
\usepackage{amsmath}
\usepackage{relsize} 
\usepackage[margin = 1.28in]{geometry}
\usepackage{color}

\usepackage{bbm}		
\usepackage{dsfont}		
\usepackage{amsthm}
\usepackage{amssymb}%
\newtheorem{theorem}{Theorem}
\newtheorem*{theorem*}{Theorem}

\newtheorem{corollary}[theorem]{Corollary}

\newtheorem*{remark*}{Remark}
\newtheorem*{assumption*}{Assumption}

\DeclareMathOperator*{\argmin}{arg\,min}

\usepackage{mathtools}

\DeclarePairedDelimiter{\floor}{\lfloor}{\rfloor}

\newtheoremstyle{dotless}{}{}{\itshape}{}{\bfseries}{}{ }{}
\theoremstyle{dotless}
  \newtheorem*{assump*}{Assumption}

\begin{document}

\title{\textsc{Identification of Noncausal Models by Quantile Autoregressions}}
\author{\textbf{Alain Hecq}
\and \textbf{Li Sun\thanks{Corresponding author: Maastricht University, School of Business and Economics,
Department of Quantitative Economics, P.O.Box 616,
6200 MD Maastricht, The Netherlands. Email: l.sun@maastrichtuniversity.nl}{\let\thefootnote\relax\footnote{We would like to thank Sean Telg, Sébastien Fries as well as the participants in the following conferences and seminar places in which we have presented our paper: The 26th Annual Symposium of the Society for Nonlinear Dynamics and Econometrics, Tokyo, March 2018, The Eighth International Conference on Mathematical and Statistical Methods for Actuarial Sciences and Finance, Madrid, April 2018, The 12th NESG 2018 Conference, Amsterdam, May 2018, The 2nd International Conference on Econometrics and Statistics, Hong Kong, June 2018, The 4th Dongbei Econometrics Workshop, Dalian, June 2018, The 2018 IAAE Annual Conference International Association for Applied Econometrics, Montreal, June 2018, The 12th International Conference on Computational and Financial Econometrics, Pisa, December 2018, The University of Milano-Bicocca, The BI Norwegian Business School and The University of Nottingham. All errors are ours.} }} }
\date{Maastricht University \\[0.38cm] \today}
%\date{December 18, 2018}
\maketitle

\begin{abstract}
We propose a model selection criterion to detect purely causal from purely noncausal models in the framework of quantile autoregressions (QAR). We also present asymptotics for the i.i.d. case with regularly varying distributed innovations in QAR. This new modelling perspective is appealing for
investigating the presence of bubbles in economic and financial time series, and is an alternative to approximate maximum likelihood methods. We illustrate our analysis using hyperinflation episodes in Latin American countries.

\bigskip



\end{abstract}

\section{Motivation}

Mixed causal and noncausal time series models have been recently used in order
(i) to obtain a stationary solution to explosive autoregressive processes,
(ii) to improve forecast accuracy, (iii) to model expectation mechanisms
implied by economic theory, (iv) to interpret non-fundamental shocks resulting
from the asymmetric information between economic agents and econometricians,
(v) to generate non-linear features from simple linear models with
non-Gaussian disturbances, (vi) to test for time reversibility. When the
distribution of innovations is known, a non-Gaussian likelihood approach can be used to discriminate between lag and lead polynomials of the dependent variable. For instance, the $\mathcal{R}$ package \texttt{MARX}
developed by Hecq, Lieb and Telg (2017) estimates univariate mixed models under the assumption of a Student's $t-$distribution with $v$ degrees of
freedom (see also Lanne and Saikkonen, 2011, 2013) as well as the Cauchy distribution as a
special case of the Student's $t$ when $v=1$. Gouri\'{e}roux and Zakoian
(2016) privilege the latter distribution to derive analytical
results. Gouri\'{e}roux and Zakoian (2015), Fries and Zakoian (2017) provide an additional flexibility to involve some skewness by using the family of alpha-stable distributions. However, all those aforementioned results require the estimation of a parametric distributional form. In this article we take another route.

The objective of this paper is to detect noncausal from causal models. To achieve that, we adopt a quantile regression (QR) framework and apply quantile autoregressions (QCAR hereafter) (Koenker and Xiao, 2006) on candidate models. Although we obviously also require non-Gaussian innovations in time series, we do not make any parametric distributional
assumption about the innovations. In quantile regressions a statistic called the sum of rescaled absolute residuals (SRAR hereafter) is used to distinguish model performances and reveal properties of time series. Remarkably we find that SRAR cannot always favour a model uniformly along quantiles. This issue is common for time series of asymmetric distributed innovations, which causes confusion in model detection and calls for a robust statistic to fit the goal. Considering that, we also propose to aggregate the SRAR information along quantiles. 

The rest of this paper is constructed as follows. Section~2 introduces mixed causal and noncausal models and our research background. In Section~3, we propose quantile autoregression in the time reverse version called quantile noncausal autoregression (QNCAR) along with a generalized asymptotic theorem in a stable law for both QCAR and QNCAR. Section~4 brings out the issue in the SRAR comparison for model detection. The use of the aggregate SRAR over all quantiles as a new model selection criterion is then proposed with the shape of SRAR curves being analysed. Furthermore, we illustrate our analysis using hyperinflation episodes of four Latin American countries in Section~\ref{sec:empirical_analysis}. Section~\ref{sec:conclusion} concludes this paper.

\section{Causal and noncausal time series models}

Brockwell and Davis introduce in their texbooks (1991, 2002) a univariate
noncausal specification as a way to rewrite an autoregressive process with
explosive roots into a process in reverse time with roots outside the unit
circle. This noncausal process possesses a stable forward looking solution
whereas the explosive autoregressive process in direct time does not. This
approach can be generalized to allow for both lead and lag polynomials. This
is the so called mixed causal-noncausal univariate autoregressive process for
$y_{t}$ that we denote MAR($r,s$)%
\begin{equation}
\pi(L)\phi(L^{-1})y_{t}=\varepsilon_{t}, \label{MAR}%
\end{equation}
where $\pi(L)=1-\pi_{1}L-...-\pi_{r}L^{r},$ $\phi(L^{-1})=1-\phi_{1}L^{-1}-...-\phi_{s}L^{-s}.$ $L$ is the usual backshift operator that creates
lags when raised to positive powers and leads when raised to negative powers,
i.e., $L^{j}y_{t}=y_{t-j}$ and $L^{-j}y_{t}=y_{t+j}$. The roots of both
polynomials are assumed to lie strictly outside the unit circle, that is $\pi(z)=0$ and
$\phi(z)=0$ for $|z|>1$ and therefore
\begin{equation}
y_{t}=\pi(L)^{-1}\phi(L^{-1})^{-1}\varepsilon_{t}=\sum\limits_{i=-\infty
}^{\infty}a_{i}\varepsilon_{t-i} 
\label{Laurent}%
\end{equation}
has an infinite two sided moving average representation. We also have that
$E(|\varepsilon_{t}|^{\delta})<\infty$ for $\delta>0$\footnote{The errors do not necessarily have finite second order moments. For $\delta \geq 2$ the second order moment exists, for $\delta\in [1,2) $ the errors have infinite variance but finite first order moment, for $\delta\in (0,1) $ the errors do not have finite order moments.} and the Laurent
expansion parameters are such that $\ \sum\limits_{i=-\infty}^{\infty}%
|a_{i}|^{\delta}<\infty.$ The representation (\ref{Laurent}) is sometimes
clearer than (\ref{MAR}) to motivate the terminology "causal/noncausal".
Indeed those terms refer to as the fact that $y_{t}$ depends on a causal
(resp. noncausal) component $\sum\limits_{i=0}^{\infty}a_{i}\varepsilon_{t-i}$
(resp. noncausal $\sum\limits_{i=-\infty}^{-1}a_{i}\varepsilon_{t-i}).$ With
this in mind, it is obvious that an autoregressive process with explosive
roots will be defined as noncausal.

Note that in (\ref{MAR}), the process $y_{t}$ is a purely causal MAR($r,0)$,
also known as the conventional causal AR($r$) process, when $\phi_{1}%
=...=\phi_{s}=0,$
\begin{equation}
\pi(L)y_{t}=\varepsilon_{t},
\label{eq:AR_r}
\end{equation}
while the process is a purely noncausal MAR($0,s)$
\begin{equation}
\phi(L^{-1})y_{t}=\varepsilon_{t},
\label{eq:MAR_0s}
\end{equation}
when $\pi_{1}=...=\pi_{r}=0.$

A crucial point of this literature is that innovation terms $\varepsilon_{t}$ must
be i.i.d. non-Gaussian to ensure the identifiability of a causal from a
noncausal specification (Breidt, Davis, Lii and Rosenblatt, 1991). The departure from
Gaussianity is not as such an ineptitude as a large part of macroeconomic and
financial time series display nonlinear and non-normal features.

We have already talked in Section 1 about the reasons for looking at models
with a lead component. Our main motivation in this paper lies in the fact that MAR$(r,s)$
models with non-Gaussian disturbances are able to replicate non-linear
features (e.g., bubbles, asymmetric cycles) that previously were usually
obtained by highly nonlinear models. As an example, we simulate in Figure~\ref{Simulation of a MAR(1,1) model, T=200} an MAR(1,1) of $(1-0.8L)(1-0.6L^{-1})y_{t}=\varepsilon_{t}$ with
$\varepsilon_{t} \overset{d}{\sim} t(3)$ for $200$ observations.\footnote{We use the
package MARX develop in $\mathcal{R}$ by Hecq, Lieb and Telg (2017).} One can
observe asymmetric cycles and multiple bubbles.%
\begin{remark*}
MAR($r,s$) models can be generated in two steps (see Gourieroux and Jasiak, 2016; Hecq,
Lieb and Telg, 2016). We propose in the Appendix an alternative method based on matrix representation that is very compact in code writing and intuitive in understanding. 
\end{remark*}

\begin{figure}[hptb]
\centering
\includegraphics[height=6cm, width=10cm]{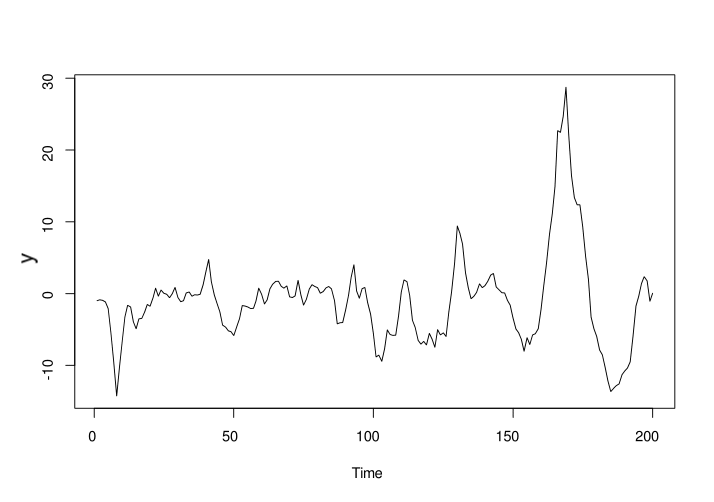}
\caption{Simulation of a MAR(1,1) model, T=200}%
\label{Simulation of a MAR(1,1) model, T=200}%
\end{figure}
%EndExpansion

Once a distribution or a group of distributions is chosen, the parameters in
$\pi(L)\phi(L^{-1})$ can be estimated.
Assuming for instance a non-standardized $t-$distribution for the innovation
process, the parameters of mixed causal-noncausal autoregressive models of the
form (\ref{MAR}) can be consistently estimated by approximate maximum
likelihood (AML). Let $(\varepsilon_{1},...,\varepsilon_{T})$ be a sequence of
i.i.d. zero mean $t-$distributed random variables, then its joint
probability density function can be characterized as
\[
f_{\varepsilon}(\varepsilon_{1},...,\varepsilon_{T}|\sigma,\nu)=\prod
_{t=1}^{T}\frac{\Gamma(\frac{\nu+1}{2})}{\Gamma(\frac{\nu}{2})\sqrt{\pi\nu
}\sigma}\left(  1+\frac{1}{\nu}\left(  \frac{\varepsilon_{t}}{\sigma}\right)
^{2}\right)  ^{-\frac{\nu+1}{2}},
\]
where $\Gamma(\cdot)$
denotes the gamma function. The corresponding (approximate) log-likelihood function conditional on the
observed data $y=(y_{1},...,y_{T})$ can be formulated as
\begin{align}
l_{y}(\boldsymbol{\phi},\boldsymbol{\varphi}%
,\boldsymbol{\lambda},\alpha|y)  &  =\left(  T-p\right)  \left[  \ln
(\Gamma((\nu+1)/2))-\ln(\sqrt{\nu\pi})-\ln(\Gamma(\nu/2))-\ln(\sigma)\right]
\nonumber\label{loglik}\\
&  -(\nu+1)/2\sum_{t=r+1}^{T-s}\ln(1+((\pi(L)\phi(L^{-1})y_{t}-\alpha
)/\sigma)^{2}/\nu),
\end{align}
where $p=r+s$ and $\varepsilon_{t}=\pi(L)\phi(L^{-1})y_{t}-\alpha$ is replaced by
a nonlinear function of the parameters when expanding the product of
polynomials. The distributional parameters are collected in
$\boldsymbol{\lambda}=[\sigma,\nu]^{\prime}$, with $\sigma$ representing the
scale parameter and $\nu$ the degrees of freedom. $\alpha$ denotes an
intercept that can be introduced in model (\ref{MAR}). Thus, the AML estimator corresponds to the
solution $\hat{\boldsymbol{\theta}}_{ML}=\arg\max_{\boldsymbol{\theta}%
\in\Theta}l_{y}(\boldsymbol{\theta}|y),$ with $\boldsymbol{\theta
}=[\boldsymbol{\phi}^{\prime},\boldsymbol{\varphi}^{\prime}%
,\boldsymbol{\lambda}^{\prime}]^{\prime}$ and $\Theta$ is a permissible parameter
space containing the true value of $\boldsymbol{\theta}$, say
$\boldsymbol{\theta}_{0}$, as an interior point. Since an analytical solution
of the score function is not directly available, gradient based numerical
procedures can be used to find $\hat{\boldsymbol{\theta}}_{ML}$. If $\nu>2$,
and hence $E(|\varepsilon_{t}|^{2})<\infty$, the AML estimator is $\sqrt{T}%
$-consistent and asymptotically normal. Lanne and Saikonen (2011) also show
that a consistent estimator of the limiting covariance matrix is obtained from
the standardized Hessian of the log-likelihood. For the estimation of the
parameters and the standard innovations as well as for the selection of mixed
causal-noncausal models we can also follow the procedure proposed by Hecq,
Lieb and Telg (2016). 

However, the AML estimation is based on a parametric form of the innovation term in~\eqref{MAR}, which makes this method not flexible enough to adapt uncommon distributions as complex in reality. To be more practical and get rid of strong distribution assumptions on innovations, in next section we adopt quantile regression methods with some properties discussed there. This paper only focuses on purely causal and noncausal models.

\section{QCAR \& QNCAR}
Koenker and Xiao (2006) have introduced a quantile autoregressive model of order $p$ denoted as QAR($p$) which is formulated as the following form:
\begin{equation}
y_{t}=\theta_{0}(u_{t})+\theta_{1}(u_{t})y_{t-1}+...+\theta_{p}(u_{t})y_{t-p},\qquad t = p+1,\ldots,T,
\label{QAR1}%
\end{equation}
where $u_{t}$ is a sequence of i.i.d. standard uniform random variables. In order to emphasize the causal characteristic of this kind of autoregressive models, we refer \eqref{QAR1} to as QCAR($p$) hereafter. Provided that the right-hand side of (\ref{QAR1}) is monotone increasing in $u_{t}%
$, the $\tau-$th conditional quantile function of $y_{t}$ can
be written as%
\begin{equation}
Q_{y_{t}}(\tau|y_{t-1},...y_{t-p})=\theta_{0}(\tau)+\theta_{1}(\tau
)y_{t-1}+...+\theta_{p}(\tau)y_{t-p}. \label{QAR2}%
\end{equation}
If an observed time series $\{y_{t}\}_{t=1}^{T}$ can be written into a QCAR($p$) process, its parameters as in (\ref{QAR2}) can be obtained from the following
minimization problem.
\begin{equation}
\hat{\boldsymbol{\theta}}(\tau)=\argmin\limits_{\boldsymbol{\theta}\in
\mathbb{R}^{p+1}}\sum\limits_{t=1}^{T}\rho_{\tau}(y_{t}-\boldsymbol{x}%
_{t}^{\prime}\boldsymbol{\theta}), \label{eq:minimization}%
\end{equation}
where $\rho_{\tau}(u):=u(\tau-I(u<0))$ is called the check function,
$\boldsymbol{x}_{t}^{\prime}:=\left[  1,y_{t-1},\ldots,y_{t-p}\right]  $, and
$\boldsymbol{\theta}^{\prime}:=\left[  \theta_{0},\theta_{1},\ldots,\theta_{t-p}\right].$ 
We define the sum of rescaled absolute residuals (SRAR) for each
pair of $(\tau,\boldsymbol{\theta})$ as
\begin{equation}
\text{SRAR}(\tau,\boldsymbol{\theta}):=\sum\limits_{t=1}^{T}\rho_{\tau}(y_{t}-\boldsymbol{x}_{t}^{\prime}\boldsymbol{\theta}). \label{eq:SRAR}%
\end{equation}
Substitute~\eqref{eq:SRAR} into~\eqref{eq:minimization}, the minimization problem~\eqref{eq:minimization} is written as 
\begin{equation}
\hat{\boldsymbol{\theta}}(\tau)=\argmin\limits_{\boldsymbol{\theta}\in
\mathbb{R}^{p+1}}\;\text{SRAR}(\tau,\boldsymbol{\theta}). \label{eq:minimization_SRAR}%
\end{equation}
The estimation consistency and asymptotic normality in the minimization problem~\eqref{eq:minimization} have been provided by Koenker and Xiao (2006). A modified simplex algorithm proposed by Barrodale and Roberts (1973) can be used to solve the minimization, and in practice
parameters for each $\tau-$th quantile can be obtained, for instance, through
the \texttt{rq()} function from the \textbf{quantreg} package in $\mathcal{R}$ or in EViews.

\subsection{QNCAR}
\label{subsec:estimation_QNCAR}
% 1) introduce QNCAR
% 2) consistency in correct specification
% 3) binding function in misspecification 
A QNCAR($p$) specification is introduced here as the noncausal counterpart of the QCAR($p$) model by reversing time, explicitly as follows:
%\footnote{Note that mixed causal noncausal models are under investigation but results are not reported at this earlier stage of our research.}:
\begin{equation}
Q_{y_{t}}(\tau|y_{t+1},...y_{t+p})=\phi_{0}(\tau)+\phi_{1}(\tau)y_{t+1}+...+\phi_{p}(\tau)y_{t+p}. \label{eq:QNCAR(p)}%
\end{equation}
Analogously to the QCAR($p$), the estimation of the QNCAR($p$) goes through solving 
$$
\hat{\boldsymbol{\theta}}(\tau)=\argmin\limits_{\boldsymbol{\theta}\in
\mathbb{R}^{p+1}}\;\text{SRAR}(\tau,\boldsymbol{\theta})
$$
\vspace{5pt}
with
$$
\boldsymbol{x}_{t}^{\prime}=\left[  1,y_{t+1},\ldots,y_{t+p}\right],
$$
where for the simplicity of the notations we use $\hat{\boldsymbol{\theta}}(\tau)$ to denote the estimate in quantile noncausal autoregression. Drawing on the asymptotics derived by Koenker and Xiao (2006), we present the following theorem for a QNCAR($p$) based on three assumptions which are made to ensure covariance stationarity of the time series (by (A1) and (A2)) and the existence of quantile estimates (by (A3)). 
\begin{remark*}
There is an issue in the estimation consistency of QCAR($p$) as reported by Fan and Fan (2010). This is due to the violation on the monotonicity requirement of the right side of \eqref{QAR1} in $u_t$ but not exclusively the monotonicity of $\theta_i(u_t)$ in $u_t$. So to recover an AR($p$) DGP of coefficients $\theta_i(u_t)\;(i=0,\ldots,p)$ monotonic in $u_t$, quantile autoregression is not a 100\% match tool unless the monotonicity requirement is met beforehand. This issue is also illustrated in Section~\ref{subsec:crossing_feature}.
\end{remark*}
 
\begin{theorem}
\label{thm:asym_QNCAR}{\ }
A QNCAR($p$) can be written in the following vectorized companion form:
\begin{equation}
\tilde{\boldsymbol{x}}_t = \mathbf{A}_t\tilde{\boldsymbol{x}}_{t+1} + \boldsymbol{\nu}_t, 
\label{eq:vector_QNCAR}
\end{equation}
where $\tilde{\boldsymbol{x}}^{\prime}_t :=\left[ y_t, y_{t+1},\ldots, y_{t+p-1}\right]$, $\boldsymbol{x}_{t}^{\prime}:=\left[  1, \tilde{\boldsymbol{x}}_t^{\prime} \right]$, $\mathbf{A}_t := \begin{bmatrix} 
\phi_{1,t} & \phi_{2,t} & \ldots & \phi_{p,t} \\
& \mathbf{I}_{p-1}  & & \boldsymbol{0}_{(p-1)\times 1}
\end{bmatrix}$
and 
$\boldsymbol{\nu}_t := \begin{bmatrix} 
\varepsilon_t \\
\boldsymbol{0}_{(p-1)\times 1}
\end{bmatrix}$,
satisfying the following assumptions:
\begin{enumerate}
\item[(A1)]: $\left\{ \varepsilon_t\right\}_{t=1}^{n}$ are i.i.d. innovations with mean $0$ and variance $\sigma^2 < \infty$. The distribution function of $\varepsilon_t$, denoted as $F(\cdot)$, has a continuous density $f(\cdot)$  with $f(\varepsilon)>0$ on $\mathcal{U}:= \left\{\varepsilon: 0< F(\varepsilon)<1 \right\}$.
\item[(A2)]: The eigenvalues of $E\left[ \mathbf{A}_t \otimes \mathbf{A}_t \right]$ have moduli less than one.
\item[(A3)]: $F_{y_t|\tilde{\boldsymbol{x}}_{t+1}}(\cdot):= \mathbf{P}\left[ y_t < \cdot \,\middle|\, y_{t+1},  y_{t+2}, \ldots, y_{t+p} \right]$ has derivative $f_{y_t|\tilde{\boldsymbol{x}}_{t+1}}(\cdot)$ which is uniformly integrable on $\mathcal{U}$ and non-zero with probability one.
\end{enumerate}
Then,
\begin{equation}
\mathbf{\Sigma}^{-\frac{1}{2}}\sqrt{T}\left(\hat{\boldsymbol{\theta}}(\tau )- \boldsymbol{\phi}(\tau) \right) \overset{d}{\sim} \mathbb{B}_{p+1}(\tau),
\label{eq:asym_qNCAR(p)}
\end{equation}
where $\mathbf{\Sigma}:= \mathbf{\Sigma}_1^{-1}\mathbf{\Sigma}_0\mathbf{\Sigma}_1^{-1}$, $\mathbf{\Sigma}_0:=E\left[ \boldsymbol{x}_{t}\boldsymbol{x}_{t}' \right]$, $\mathbf{\Sigma}_1:=\lim T^{-1}\sum^T_{t=1} f_{y_t|\tilde{\boldsymbol{x}}_{t+1}}\left( F_{y_t|\tilde{\boldsymbol{x}}_{t+1}}^{-1}\bigl(\tau\bigr)\right)\,\boldsymbol{x}_{t}\boldsymbol{x}_{t}^{\prime} $, $\boldsymbol{\phi}(\tau)':= \left[F^{-1}(\tau), \phi_1(\tau), \ldots, \phi_p(\tau) \right]$, $\mathbb{B}_{p+1}(\tau):= \mathcal{N}\left( \boldsymbol{0},\tau(1-\tau)\mathbf{I}_{p+1} \right)$ with sample size $T$.
\end{theorem} 
The above result can be further simplified into Corollary~\ref{corollary_constantcoef} by adding the following assumption:
\begin{enumerate}
\item[(A4):]
\textit{The coefficient matrix $\mathbf{A}_t$ in~\eqref{eq:vector_QNCAR} is constant over time. (We denote $\mathbf{A} := \begin{bmatrix} 
\phi_{1} & \phi_{2} & \ldots & \phi_{p} \\
& \mathbf{I}_{p-1}  & & \boldsymbol{0}_{(p-1)\times 1}
\end{bmatrix}$} for $\mathbf{A}_t$  under this assumption.)
\end{enumerate} 
\begin{corollary}
\label{corollary_constantcoef}
Under assumptions (A1), (A2), (A3) and (A4), 
\begin{equation}
\sqrt{T}\,f\left( F^{-1}\bigl(\tau\bigr)\right)\,\mathbf{\Sigma}_0^{\frac{1}{2}}\left(\hat{\boldsymbol{\theta}}(\tau )- \boldsymbol{\phi}_{\tau}\right) \overset{d}{\sim} \mathbb{B}_{p+1}(\tau),
\end{equation}
where $\boldsymbol{\phi}_{\tau} := \left[F^{-1}(\tau), \phi_1, \ldots, \phi_p \right]$.
\end{corollary}

As can be seen, QCAR($p$) and QNCAR($p$) generalize the classical purely causal and purely noncausal models respectively by allowing random coefficients on lag or lead regressors over time. Corollary~2 provides additional results when the same coefficients except the intercept are used to generate each quantile. However, the moment requirement in (A1) is very strict for heavy tailed time series. In order to study noncausality by QAR in heavy tailed distributions, we have to show its applicability without assumption (A1). This goal is achieved by Theorem~\ref{thm:asym_regvary_NCAR} which presents the asymptotic behaviour of the QAR estimator for a classical purely noncausal model. Similarly, the asymptotics for a classical purely causal model follows right after reversing time.
\\

\begin{theorem}[Asymptotics in regularly varying distributed innovations] \label{thm:asym_regvary_NCAR}{\ } \\ 
Under Assumption (A4), a purely noncausal AR(p) of the following form
$$
\phi(L^{-1})y_{t}=\varepsilon_{t},
$$
where $\phi(L^{-1})=1-\phi_{1}L^{-1}-...-\phi_{p}L^{-p}$, also
satisfies the following assumptions:
\begin{enumerate}
\item[(A5)]: $\left\{ \varepsilon_t\right\}_{t=1}^{n}$ are i.i.d. innovation variables with regularly vary tails defined as
\begin{equation}
\label{eq:regular_vary_tail}
P\left(\left|\varepsilon_t\right|>x\right) = x^{-\alpha}L(x),
\end{equation}
where $L(x)$ is slowly varying at $\infty$ and $0<\alpha<2$.
There is a sequence $\left\{a_T\right\}$ satisfying
\begin{equation}
T \cdot P\left\{ |\varepsilon_t|>a_T\,x \right\} \rightarrow x^{-\alpha} \qquad \text{for all}\; x>0.
\label{eq:aT}
\end{equation} 
with $b_T = \mathbb{E}\left[ \varepsilon_t\,I[|\varepsilon_t|\leq a_T]\right]= 0$.\footnote{Without loss of generality, we assume $b_T$ to be zero in the derivation for the simplicity.} The distribution function of $\varepsilon_t$, denoted as $F(\cdot)$, has continuous density $f(\cdot)$  with $f(\varepsilon)>0$ on $\left\{\varepsilon: 0< F(\varepsilon)<1 \right\}$ in probability one;
\item[(A6)]: The roots of the polynomial $\phi(z)$ are greater than one, such that $y_t$ can be written into
\begin{equation}
	y_t = \sum\limits_{j=0}^{\infty}c_j\,\varepsilon_{t+j}, 
\end{equation}
where $\sum\limits_{j=0}^{\infty} j\,|c_j|^{\delta} < \infty $ for some $\delta<\alpha, \delta\leq 1$.
\end{enumerate}
Then
\begin{equation}
\begin{aligned}
	\frac{f\left( F^{-1}(\tau)\right)\cdot a_T\sqrt{T}}{\sqrt{\tau(1-\tau)}\,}  & \left(\hat{\boldsymbol{\theta}}(\tau ) - \boldsymbol{\phi}_{\tau} \right)   \overset{d}{\sim} \qquad\qquad	\\
 \begin{bmatrix} 
1 		& \boldsymbol{0} 	\\
\boldsymbol{0}		& \Omega_1^{-1}\Omega_{\boldsymbol{S}}^{-1}
\end{bmatrix} & \left[ W(1), \sum\limits^{\infty}_{j=0} c_j\, \int^1_0 \mathcal{S}_{\alpha}(s)\,dW(s)\,, \ldots, \sum\limits^{\infty}_{j=0} c_j\, \int^1_0 \mathcal{S}_{\alpha}(s+\frac{p-1}{T})\,dW(s) \right]_{(p+1)\times 1}.
\end{aligned}
\label{eq:asym_qNCAR(p)}
\end{equation}
where $\boldsymbol{\phi}_{\tau} := \left[\frac{F^{-1}(\tau)}{a_T}, \phi_1, \ldots, \phi_p  \right]$, 
%$\Omega : = 
%\begin{bmatrix} 
%1 		& \boldsymbol{0} 	\\
%\boldsymbol{0}		& \Omega_1
%\end{bmatrix}_{(p+1)\times (p+1)}$,
$\Omega_1$ being a $p\times p$ matrix with entry $\omega_{ik}:= \sum\limits^{\infty}_{j=0}  c_j\, c_{j+|k-i|}$ at the $i$-th row and the $k$-th column
, $\left\{\mathrm{S}_{\alpha}(s)\right\}$ being a process of stable distributions with index $\alpha$ which is independent of Brownian motion $\left\{W(s) \right\}$, and $\Omega_{\boldsymbol{S} }$ being a $p\times p$ diagonal matrix with the $j-$th diagonal entry being $ \int_0^1 \mathrm{S}_{\alpha}^2(s + \frac{j-1}{p})\,ds $, $j\in\left\{1,2,\ldots,p \right\}$. In this theorem the intercept regressor in QNCAR($p$) is changed to $a_T$ so that $\boldsymbol{x}^{\prime}_t :=\left[a_T, y_t, y_{t+1},\ldots, y_{t+p-1}\right]$.
%asymptotically in a stable distribution.
\end{theorem} 

\begin{proof}[\textbf{Proof.}]
See the appendix.
\end{proof}

Heuristically, next we restrict our focus on the classical models and explore consequences of causality misspecification in quantile regressions.  
% from two aspects: residuals and estimation
%conceived after comparing the residuals by misspecification and by correct specification in quantile regression.   

\subsection{Causal and noncausal models with Gaussian i.i.d. disturbances}

Suppose a causal AR(1) process $\left\{  y_{t}\right\}_{t=1}^T $, $y_{t}=\alpha+\beta y_{t-1}+\varepsilon_{t}$, with for instance $\left[  \alpha,\beta\right]  =[1,0.5]$, i.i.d. standard normal distributed $\left\{  \varepsilon_{t}\right\}  $ and $T=200.$ Figure~\ref{fig:series_G_C} displays a corresponding simulated series.
\begin{figure}[hptb]
\centering
\includegraphics[height=8cm, width=10cm]%
{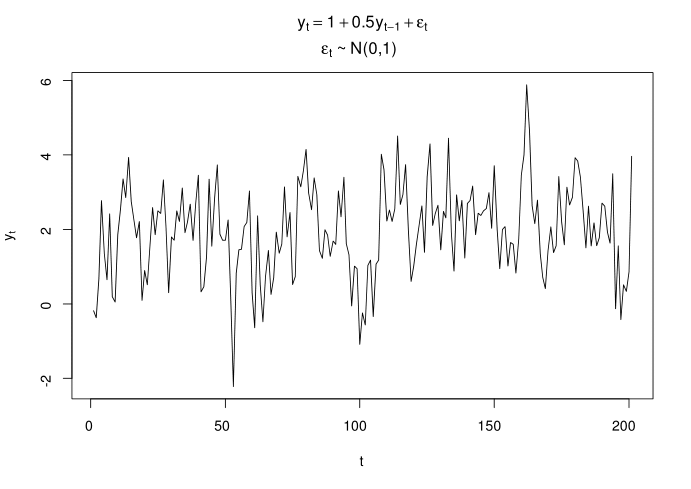}%
\caption{Simulation of a one-regime process with $N(0,1)$ innovations, $T=200$}%
\label{fig:series_G_C}
\end{figure}

The information displayed in Figure~\ref{fig:SRAR_G_C} is the SRAR$(\tau)$ of each candidate model along quantiles, indicating their goodness of fit. The two SRAR curves almost overlap at every quantile, which implies no discrimination between QCAR and QNCAR in Gaussian innovations, in line with results in the OLS case. The Gaussian distribution is indeed time reversible, weak and strict stationary. Its first two moments characterize the whole
distribution and consequently every quantile. Note that we obtain similar results for a stationary noncausal AR($p$) process with i.i.d. Gaussian $\left\{  \varepsilon_{t}\right\}  $.
The results are not reported to save space.%
%TCIMACRO{\FRAME{ftbpFO}{3.8885in}{2.7523in}{0pt}{\Qct{Sum of rescaled absolute
%residuals along quantiles for an AR(1) with $N(0,1)$ innovations, $T=200$}}%
%{}{srar_g_c.eps}{\special{ language "Scientific Word";  type "GRAPHIC";
%width 3.8885in;  height 2.7523in;  depth 0pt;  original-width 7.862in;
%original-height 5.5486in;  cropleft "0";  croptop "1";  cropright "1";
%filename 'C:/Users/Li Sun/Dropbox/Thesis_Li/1st Thesis/QAR simulation/report writing/epsgraphs/T200/SRAR_G_C.eps';file-properties "XNPEU";}%
%} }%
%BeginExpansion
\begin{figure}[hptb]
\centering
\includegraphics[height=8cm, width=10cm]%
{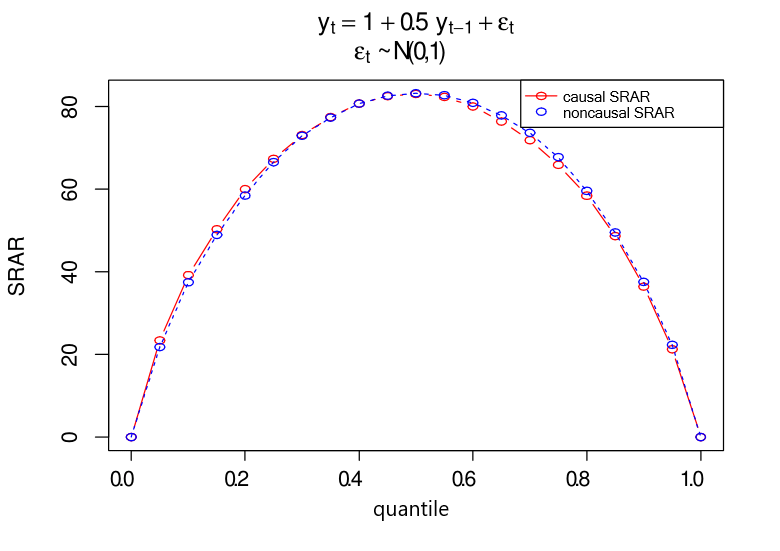}%
\caption{SRAR plot under an AR(1) with $N(0,1)$ innovations, $T=200$}
\label{fig:SRAR_G_C}
\end{figure}

\subsection{Causal and noncausal models with Student's $t$ distributed innovations}

Things become different if we depart from Gaussianity. Suppose now a causal AR(1)
process $y_{t}=\alpha+\beta y_{t-1}+\varepsilon_{t}$ with again $\left[\alpha,\beta\right]  =[1,0.5]$ but where $\left\{\varepsilon_{t} \right\}$
are i.i.d. Student's $t-$distributed with 2 degrees of freedom (hereafter using shorthand notation: $t(2)$). Figure~\ref{fig:series_t2_C}
depicts a simulation in this AR(1) with $T=200$. Applying QCAR and QNCAR respectively on this series results in the SRAR
curves displayed in Figure~\ref{fig:SRAR_t2_NC}. The distance between the two curves is obvious compared to the Gaussian case, favouring the causal specification at almost all quantiles. Figure~\ref{fig:SRAR_t1_NC} is the SRAR plot of a purely noncausal process with
i.i.d. Cauchy innovations. The noncausal
specification is preferred in the SRAR comparison.

It seems now that applying the SRAR comparison at one quantile, such as the median, is sufficient for model identification, but it is not true in general. In Section~\ref{sec:SRAR}, we will observe a crossing feature in SRAR plots, the true model even having higher SRAR values at certain quantiles than the misspecified model.
\begin{figure}[hptb]
\centering
\includegraphics[height=6cm, width=10cm]%
{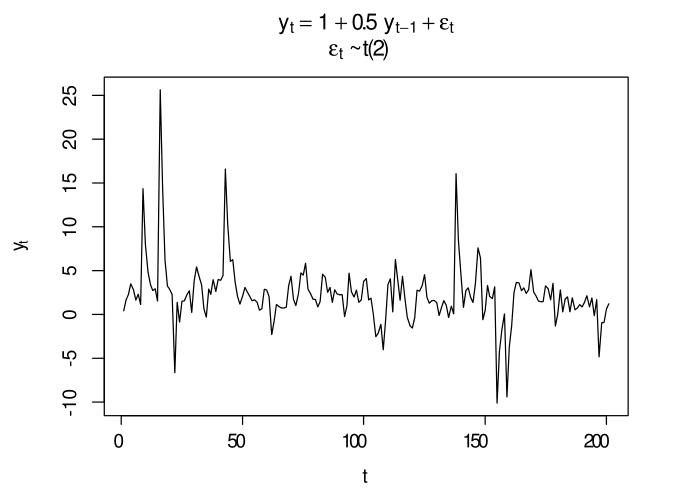}%
\caption{Simulation of an AR(1) with $t$(2) innovations, $T=200$}%
\label{fig:series_t2_C}
\end{figure}

\begin{figure}[hptb]
\centering
\includegraphics[height=6cm, width=10cm]%
{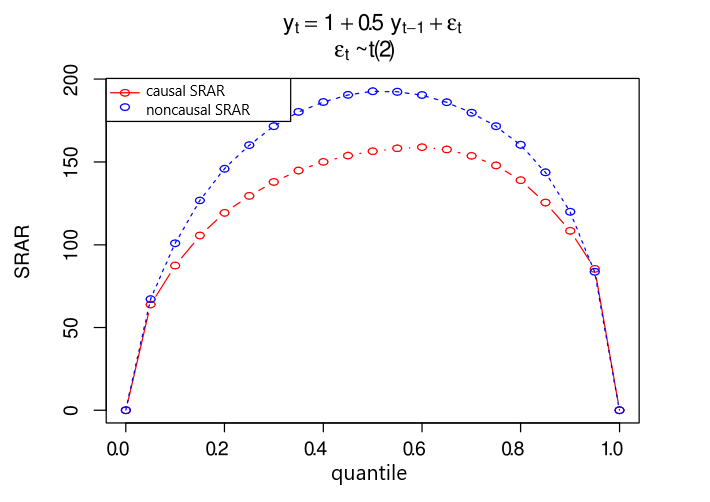}%
\caption{SRAR plot under an AR(1) with $t{(2)},$ $T=200$}%
\label{fig:SRAR_t2_NC}
\end{figure}

\begin{figure}[hptb]
\centering
\includegraphics[height=8cm, width=10cm]%
{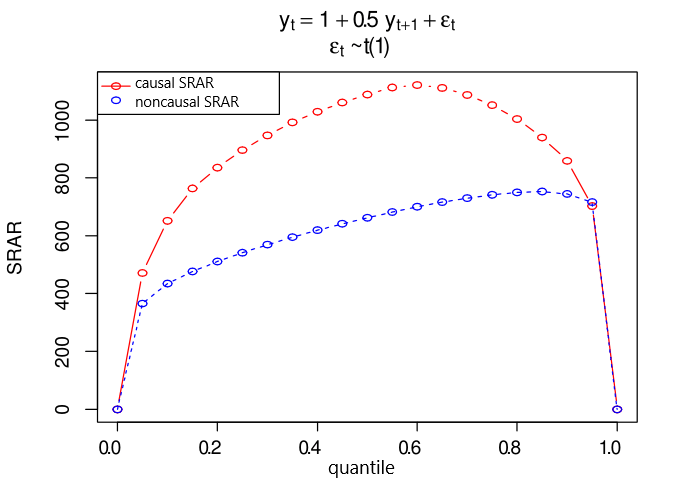}%
\caption{SRAR plot under a noncausal model with Cauchy innovations, $T=200$}%
\label{fig:SRAR_t1_NC}
\end{figure}

\begin{figure}[hptb]
\centering
\includegraphics[height=8cm, width=10cm]%
{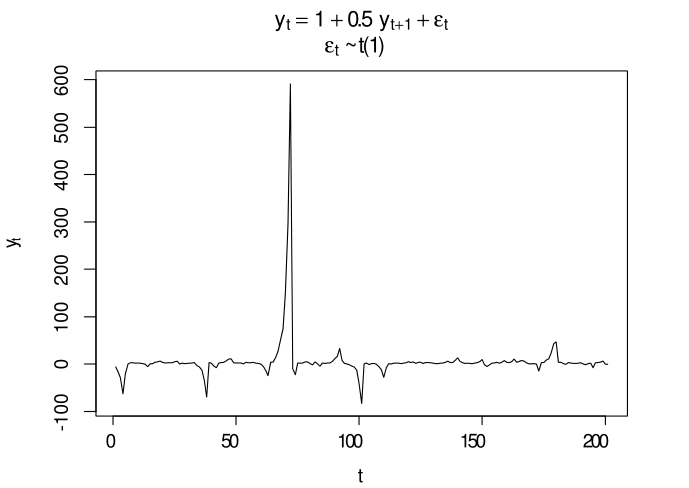}%
\caption{Simulation of a noncausal model with Cauchy innovations, $T=200$}%
\label{fig:series_t1_NC}
\end{figure}

So far we have applied QCAR and its extension QNCAR for purely causal or noncausal models with symmetrically i.i.d. innovation series. We show that a model selection by the SRAR comparison gives uniform decisions along quantiles. However, such a model selection is not always that clear in practice. In the empirical study, we will encounter a crossing phenomenon in SRAR plots. In the next section, we will present such a crossing phenomenon with some possible reasons, and propose a more robust model selection criterion called the aggregate SRAR.

\section{SRAR as a model selection criterion}
\label{sec:SRAR}
It is natural to think about SRAR as a model selection criterion since a lower SRAR means a better goodness of fit in quantile regressions. However, SRAR is a function of quantile, which raises a question on which quantile to be considered for model selection. It is empirically common to see a crossing feature of SRAR plots, which gives different model selections at certain quantiles and makes a selection unreliable if only one quantile is considered. In this section, we discuss this issue and propose a more robust model selection criterion based on aggregating SRARs.

\subsection{Crossing feature of SRAR plots}
\label{subsec:crossing_feature}
First let us see some possible model settings causing crossings in SRAR plots. The first case is linked to the existence of multi-regimes in coefficients.

Suppose a QNCAR($p$) process specified as follows: 
\begin{equation}
\left\{
\begin{array}
[c]{c}%
y_{t}=\beta_{1}y_{t+1} + F^{-1}\left( \tau_{t} \right)\text{ \ \ \ \ if }0\leq\tau_{t}\leq\tau^{\ast}\\
y_{t}=\beta_{2}y_{t+1} + F^{-1}\left( \tau_{t} \right)\text{ \ \ \ \ if }\tau^{\ast}<\tau
_{t} \leq 1
\end{array}
\right.  \label{eq:two_regime_QNCAR}%
\end{equation}
where $\left\{ \tau_{t}\right\} $ is a sequence of i.i.d. standard uniform random variables, equating $y_t$ with its $\tau_t\,$th conditional quantile if and only if the right-hand side is monotonically increasing in $\tau_t$. $F(\cdot)$ is the cumulative density function of the i.i.d. innovation process. There is a problem in using quantile autoregression to recover coefficients in this model if the monotonicity requirement of the right side in $\tau_t$ is violated. Because of the violation, this regime model at $\tau$-th regime is no longer in coincidence with its $\tau$-th conditional quantile. This makes the QNCAR unable to recover the true regime model. However, if the random coefficients are monotonically increasing in $\tau_t$, then by restricting to the non-negative region of $y_{t+1}$ (also see Fan and Fan, 2010) we force this regime model in regression to satisfy the monotonicity requirement without losing its characteristics. We can then obtain the estimation consistently of the true parameters in \eqref{eq:two_regime_QNCAR}. Such a restricted QCAR (or QNCAR) is called here restricted quantile causal autoregression (or restricted quantile noncausal autoregression, RQCAR or RQNCAR hereafter), is formulated as follows:
\begin{equation}
\hat{\boldsymbol{\theta}}(\tau ) = \argmin\limits_{\boldsymbol{\theta}\in
\mathbb{R}^{p+1}}\;\sum\limits_{t=1}^{T}I\left[ t\in\mathfrak{T}\right]\rho_{\tau}%
(y_{t}-\boldsymbol{x}_{t}^{\prime}\boldsymbol{\theta})
\label{eq:RQCAR}
\end{equation}
where $\mathfrak{T}$ is the set restricting the quantile regression on particular observations. In this paper, we restrict the QNCAR on non-negative regressors, i.e., $\mathfrak{T}= \left\{t: \boldsymbol{x}_{t}\geq \boldsymbol{0} \right\}$.
\\
Figure~\ref{fig:crossing_SRAR_CQAR} shows four SRAR curves estimated from QCAR, QNCAR, RQCAR and RQNCAR. We consider a time series $\left\{ y_t \right\}_{t=1}^{600}$ simulated from the model~\eqref{eq:two_regime_QNCAR} with $\tau^{\ast} = 0.7, \beta_1 = 0.2, \beta_2 = 0.8$
and i.i.d. innovation process following a $t(3)$, i.e., $F^{-1}(\cdot)= F^{-1}_{t(3)}(\cdot)$. 
\begin{figure}[hptb]
\centering
\includegraphics[height = 8cm, width=12cm]{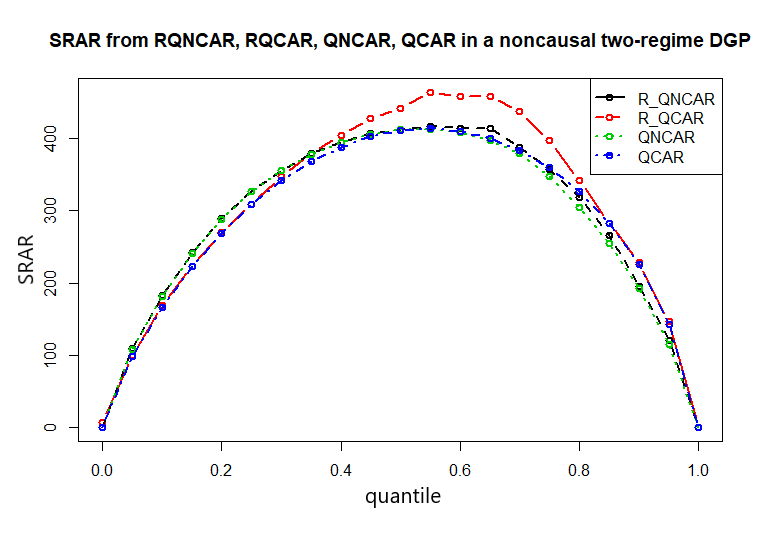}%
\caption{Crossing feature in a SRAR plot with restricted quantile autoregressions}
\label{fig:crossing_SRAR_CQAR}
\end{figure}
Figure~\ref{fig:crossing_SRAR_CQAR} illustrates such a crossing phenomenon in which the SRAR curve from a true model is not always lower than one from misspecification. Applying restriction helps to enlarge the SRAR difference between a true model and a misspecified time direction.

The second case we investigate is the presence of skewed distributed disturbances.

Let us consider a time series $\left\{ y_t \right\}$ following a purely noncausal AR($1$): $y_t = 0.8 y_{t+1} + \varepsilon_t$ with $\left\{\varepsilon_t \right\}$ i.i.d. in a demeaned skewed $t$ distribution with skewing parameter $\gamma=2$ and $v = 3$ degrees of freedom (hereafter $t(v,\gamma)$ is the shorthand notation for a skewed t-distribution). The probability density function of $t(v,\gamma)$ (see Francq and Zakoïan 2007) is defined as
\begin{equation}
\left\{
\begin{aligned}
f(x) & = \frac{2}{\gamma + \frac{1}{\gamma}}f_t(\gamma x) & \qquad \text{for}\; x<0	\\
f(x) & = \frac{2}{\gamma + \frac{1}{\gamma}}f_t(\frac{x}{\gamma}) & \qquad \text{for}\; x\geq 0
\end{aligned}
\right.
\label{eq:skewed_t_pdf}
\end{equation}
where $f_t(\cdot)$ is the probability density function of the symmetric $t(v)$ distribution. Figure~\ref{fig:crossing_SRAR_skewedT} shows four SRAR curves derived from the estimation of the QCAR, the QNCAR, the RQCAR and the RQNCAR respectively. The curves from the QNCAR and the RQNCAR almost overlap each other, which confirms our understanding that the monotonicity requirement is met in the true model. The estimations and the corresponding SRAR curves should be the same unless many observations are omitted by the restriction. On the other hand, the SRAR curve gets enlarged from the QCAR to the RQCAR, which is very reasonable as the feasible set in the QCAR is larger and the misspecification is not ensured to satisfy the monotonicity requirement. Again we see this crossing feature in the SRAR plot. Remarkably, the SRAR curve from a true model can be higher at certain quantiles than the one from a misspecified model. Consequently the SRAR comparison relying only on particular quantiles, such as the least absolute deviation (LAD) method for the median only, is not robust in general. Therefore, we propose a new model selection criterion in next subsection by including the information over all quantiles.
\begin{figure}[hptb]
\centering
\includegraphics[height = 6cm, width=10cm]{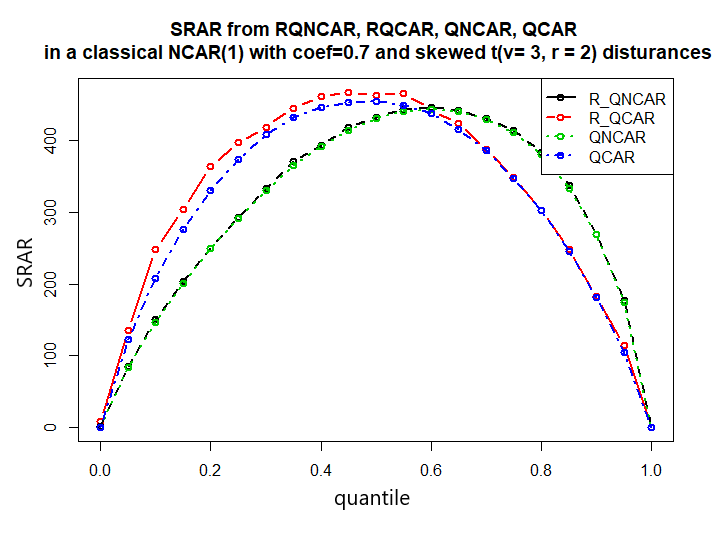}%
\caption{Crossing feature in a SRAR plot with a skewed distribution}
\label{fig:crossing_SRAR_skewedT}
\end{figure}

\subsection{The aggregate SRAR criterion}
Based on the same number of explanatory variables in QCAR and QNCAR with a fixed sample size in quantile regressions, the best model is supposed to exhibit the highest goodness of fit among candidate models. Similarly to the R-squared criterion in the OLS, when turning to quantile regressions, we are led to use a SRAR criterion for model selection. The aggregate SRAR is regarded as an overall performance of a candidate model over every quantile such as: 
\[
\text{aggregate SRAR} := \int_{0}^{1}SRAR(\tau)d\tau.
\]
There are many ways to calculate this integral. One way is to approximate the integral by the trapezoidal rule. Another way is to sum up SRARs over a fine enough quantile grid with equal weights. In other words, this aggregation is regarded as an average of performances (SRAR($\tau$), $\tau\in(0,1)$) of a candidate model. In practice, there is almost no difference in model selection between the two aggregation methods.

Performances of the SRAR model selection criteria in Monte Carlo simulations are reported in Table~\ref{tab:SRAR_compare}. It shows the frequencies with which we find the correct model based on the SRAR criterion per quantile and the aggregate SRAR criterion. The sample size $T$ is 200 and each reported number is based on 2000 Monte Carlo simulations. Columns of Table~\ref{tab:SRAR_compare} refer to as a particular distribution previously illustrated in this paper. As observed, the aggregate SRAR criterion performs very well even in crossing situations. The Gaussian distribution being weakly and strictly stationary we cannot obviously discriminate between causal and noncausal specifications leading to a frequency of around 50\% to detect the correct model.

\begin{table}[h]
\caption{Frequencies of selecting the correct model using the SRAR criteria}\vspace{5pt}%
\label{tab:SRAR_compare}%
\centering%
{\small
\begin{tabular}
[c]{lrrrrr}\hline
Quantiles & Gaussian  & $t(2)$ & $t(1)$  & two-regime  & $t(v=3,\gamma=2)$  \\
	 & (Fig.~\ref{fig:series_G_C})	& (Fig.~\ref{fig:series_t2_C}) 	& (Fig.~\ref{fig:series_t1_NC})	& (Fig.~\ref{fig:crossing_SRAR_CQAR})	& (Fig.~\ref{fig:crossing_SRAR_skewedT})	\\ \hline
0 	 			& 0.698 & 0.678 & 0.601 &	0.787 & 0.476\\
0.05 			& 0.516 & 0.416 & 0.653 &	0.044 & 1.000\\
0.10 			& 0.51  & 0.677 & 0.763 &	0.059 & 1.000\\
0.15 			& 0.519 & 0.858 & 0.841 &	0.095 & 1.000\\
0.20 			& 0.512 & 0.948 & 0.907 &	0.167 & 1.000\\
0.25 			& 0.513 & 0.981 & 0.947 &	0.305 & 1.000\\
0.30 			& 0.488 & 0.992 & 0.978 &	0.487 & 1.000\\
0.35 			& 0.487 & 0.998 & 0.996 &	0.654 & 1.000\\
0.40 			& 0.486 & 0.999 & 0.996 &	0.798 & 1.000\\
0.45 			& 0.487 & 1.000 & 0.996 &	0.892 & 1.000\\
0.50 			& 0.5 	& 1.000 & 0.995 &	0.950 & 1.000\\
0.55 			& 0.499 & 0.999 & 0.995 &	0.974 & 0.994\\
0.60 			& 0.492 & 0.999 & 0.995 &	0.988 & 0.533\\
0.65 			& 0.478 & 0.997 & 0.995 &	0.991 & 0.018\\
0.70 			& 0.467 & 0.994 & 0.979 &	0.996 & 0.001\\
0.75 			& 0.49  & 0.984 & 0.951 &	0.998 & 0.000\\
0.80 			& 0.493 & 0.954 & 0.903 &	0.999 & 0.000\\
0.85 			& 0.481 & 0.862 & 0.858 &	1.000 & 0.000\\
0.90 			& 0.469 & 0.72  & 0.791 &	1.000 & 0.000\\
0.95 			& 0.484 & 0.454 & 0.668 &	0.997 & 0.000\\
1 	 			& 0.653 & 0.58  & 0.595 &	0.780 & 0.420\\
aggregate SRAR 	& 0.483 & 0.998 & 0.995 &	0.995 & 0.999\\\hline
\end{tabular}
}
\end{table}

\subsection{Shape of SRAR curves}
By observing SRAR plots, we see that SRAR curves vary when the underling distribution varies. It is interesting to investigate the reasons. In this subsection, we will provide some insights on the slope and concavity of $\text{SRAR}_{y_t}(\tau, \hat{\boldsymbol{\theta}}(\tau))$ curves under assumptions (A1), (A2), (A3) and (A4). Since $\rho_{\tau}(y_{t}-\boldsymbol{x}_{t}^{\prime}\boldsymbol{\theta})$ is a continuous function in $\boldsymbol{\theta}\in \mathbb{R}^{(p+1)}$, by the continuous mapping theorem and $\hat{\boldsymbol{\theta}}(\tau)) \overset{p}{\rightarrow}  \boldsymbol{\phi}_{\tau}$, we know that
$$
\rho_{\tau}(y_{t}-\boldsymbol{x}_{t}^{\prime}\hat{\boldsymbol{\theta}})\overset{p}{\rightarrow}\rho_{\tau}(y_{t}-\boldsymbol{x}_{t}^{\prime}\boldsymbol{\phi}_{\tau}).
$$
%\text{SRAR}_{y_t}(\tau, \hat{\boldsymbol{\theta}}(\tau)) \overset{p}{\rightarrow} \text{SRAR}_{y_t}(\tau, \boldsymbol{\phi}_{\tau}).
We also know that
$$
\rho_{\tau}(y_{t}-\boldsymbol{x}_{t}^{\prime}\boldsymbol{\phi}_{\tau}) = \rho_{\tau}(\varepsilon_t - F^{-1}(\tau)).
$$ 
%where
%\begin{equation}
%\text{SRAR}_{\varepsilon_t}(\tau, F^{-1}(\tau)) = \sum^T_{t=1}\rho_{\tau}(\varepsilon_t - F^{-1}(\tau)).
%\end{equation}
Therefore instead of directly deriving the shape of a $\text{SRAR}_{y_t}(\tau, \hat{\boldsymbol{\theta}}(\tau))$ curve, we look at the properties of its intrinsic curve $\text{SRAR}_{\varepsilon_t}(\tau, F^{-1}(\tau))$. We derive the first and second order derivatives of $\text{SRAR}_{\varepsilon_t}(\tau, F^{-1}(\tau))$ with respect to $\tau$ in order to determine the shape of $\text{SRAR}_{y_t}(\tau, \hat{\boldsymbol{\theta}}(\tau))$.

\subsubsection{The slope property}
One major difference between SRAR curves in a plot is their slopes. We can compute the first-order derivative of SRAR with respect to $\tau$ if the derivative exists. Under the following assumption:
\vspace*{-1.5mm}
\begin{enumerate}
\item[(A7):] The inverse distribution function $F^{-1}(\cdot)$ of innovation $\varepsilon_t$ is continuous and differentiable on $(0,1)$ to the second order;
% , i.e., $F^{-1}(\cdot)\in \mathcal{C}^1(0,1)$,
\end{enumerate}
\vspace*{-1.5mm}
we can then take the first-order derivative of $\text{SRAR}_{\varepsilon_t}(\tau, F^{-1}(\tau))$ with respect to $\tau$. 
\\
Suppose $0< \tau < \tau + \Delta\tau < 1, \Delta\tau > 0$ and denote $\Delta F^{-1}(\tau) := F^{-1}(\tau + \Delta\tau) - F^{-1}(\tau)$.
\begin{equation}
\begin{aligned}
&\text{SRAR}_{\varepsilon_t}(\tau + \Delta\tau  , F^{-1}(\tau + \Delta\tau)) - \text{SRAR}_{\varepsilon_t}(\tau, F^{-1}(\tau)) 
					\\
					& = \sum^{T}_{t=1}\biggl(\rho_{\tau + \Delta\tau}\left( \varepsilon_t - F^{-1}(\tau + \Delta\tau)\right)  - \rho_{\tau}\left( \varepsilon_t - F^{-1}(\tau) \right)\biggr)
					\\
					& = \sum^{T}_{t=1}\biggl(
					 \left( \varepsilon_t - F^{-1}(\tau + \Delta\tau)\right)\left(\tau + \Delta\tau - \mathds{1}_{\left\{ \varepsilon_t - F^{-1}(\tau + \Delta\tau) \leq 0 \right\}} \right) \biggr. 
				 	 - \biggl.
					  \left( \varepsilon_t - F^{-1}(\tau)\right)\left(\tau  - \mathds{1}_{\left\{ \varepsilon_t - F^{-1}(\tau) \leq 0  \right\} } \right) 			
					\biggr)			\\
					& = \sum^{T}_{t=1}\biggl(
					  \varepsilon_t \left( \Delta\tau - \mathds{1}_{\left\{ F^{-1}(\tau ) < \varepsilon_t \leq F^{-1}(\tau + \Delta\tau) \right\} } \right)			  
					   + \tau\left( F^{-1}(\tau ) - F^{-1}(\tau+\Delta\tau) \right)							 - \Delta\tau\,F^{-1}(\tau+\Delta\tau)				\\
					  & + F^{-1}(\tau+\Delta\tau)\,\mathds{1}_{\left\{\varepsilon_t \leq F^{-1}(\tau + \Delta\tau)\right\} } 
					   - F^{-1}(\tau)\,\mathds{1}_{\left\{\varepsilon_t \leq F^{-1}(\tau ) \right\}}
					\biggr)		\\
%					& = \sum^{T}_{t=1}\biggl(
%					  \varepsilon_t \left( \Delta\tau - \mathds{1}_{\left\{ F^{-1}(\tau ) \leq \varepsilon_t < F^{-1}(\tau + \Delta\tau) \right\} } \right)
%					   + \tau\left( F^{-1}(\tau ) - F^{-1}(\tau+\Delta\tau) \right)		 - \Delta\tau\,F^{-1}(\tau+\Delta\tau)	\\
%					  & + \left( F^{-1}(\tau+\Delta\tau) - F^{-1}(\tau) \right) \, \mathds{1}_{\left\{\varepsilon_t < F^{-1}(\tau + \Delta\tau)\right\} } 
%					   + F^{-1}(\tau)\,\mathds{1}_{\left\{ F^{-1}(\tau ) \leq \varepsilon_t < F^{-1}(\tau + \Delta\tau) \right\} }
%					\biggr) 
%					\\
					& = \sum^{T}_{t=1}\biggl(
					  \Delta\tau \left( \varepsilon_t - F^{-1}(\tau+\Delta\tau) \right)
					   + \left( F^{-1}(\tau+\Delta\tau) - F^{-1}(\tau) \right) \, \left( \mathds{1}_{\left\{\varepsilon_t \leq F^{-1}(\tau + \Delta\tau)\right\} } - \tau \right) 
					   \\
					   & + \mathds{1}_{\left\{ F^{-1}(\tau ) < \varepsilon_t \leq F^{-1}(\tau + \Delta\tau) \right\} } \left( F^{-1}(\tau) - \varepsilon_t \right)
					\biggr).
\end{aligned}
\end{equation}
Divide the above difference by $\Delta\tau$, and take the limit $\Delta\tau\downarrow 0$. It gives us
\begin{equation}
\begin{aligned}
\lim\limits_{\Delta\tau\downarrow 0}& \frac{\text{SRAR}_{\varepsilon_t}(\tau + \Delta\tau  , F^{-1}(\tau + \Delta\tau)) - \text{SRAR}_{\varepsilon_t}(\tau, F^{-1}(\tau)) }{\Delta\tau} 
						\\
					& = \sum^{T}_{t=1}\biggl(
					   \varepsilon_t - F^{-1}(\tau)
					   + \frac{d F^{-1}(\tau)}{d \tau} \, \left( \mathds{1}_{\left\{\varepsilon_t \leq F^{-1}(\tau)\right\} } - \tau \right) 
					\biggr),
\end{aligned}
\end{equation}
because
\begin{equation}
\label{eq:slope_formula}
\begin{aligned}
\lim\limits_{\Delta\tau\downarrow 0}& \frac{\Delta\tau \left( \varepsilon_t - F^{-1}(\tau+\Delta\tau) \right)}{\Delta\tau} 
					 =  \varepsilon_t - F^{-1}(\tau),
					 \\
\lim\limits_{\Delta\tau\downarrow 0}& \frac{\left( F^{-1}(\tau+\Delta\tau) - F^{-1}(\tau) \right) \, \left( \mathds{1}_{\left\{\varepsilon_t \leq F^{-1}(\tau + \Delta\tau)\right\} } - \tau \right) }{\Delta\tau} 
					 =  \frac{d F^{-1}(\tau)}{d \tau} \, \left( \mathds{1}_{\left\{\varepsilon_t \leq F^{-1}(\tau)\right\} } - \tau \right) ,	
						\\
						\lim\limits_{\Delta\tau\downarrow 0}& \frac{\mathds{1}_{\left\{ F^{-1}(\tau ) < \varepsilon_t \leq F^{-1}(\tau + \Delta\tau) \right\} } \left( F^{-1}(\tau) - \varepsilon_t \right) }{\Delta\tau} 
					 =  0 .				 				 
\end{aligned}
\end{equation}
The last line is from 
\begin{equation}
\left\{
\begin{aligned}
\mathds{1}_{\left\{ F^{-1}(\tau ) < \varepsilon_t \leq F^{-1}(\tau + \Delta\tau) \right\} } \left( F^{-1}(\tau) - \varepsilon_t \right)  = 0, \qquad 
				&\text{when}\quad \varepsilon_t \not\in \bigl( F^{-1}(\tau ), F^{-1}(\tau + \Delta\tau) \bigr]; 
						\\
 \left( F^{-1}(\tau) - F^{-1}(\tau + \Delta\tau) \right) \leq \left( F^{-1}(\tau) - \varepsilon_t \right)  < 0  , \quad  & \text{when}\quad \varepsilon_t \in \bigl( F^{-1}(\tau ), F^{-1}(\tau + \Delta\tau) \bigr];
\end{aligned}
\right.
\end{equation}
and
\begin{equation}
				 	 0=\mathds{1}_{\left\{ F^{-1}(\tau ) < \varepsilon_t \leq F^{-1}(\tau) \right\} }\frac{d F^{-1}(\tau)}{d \tau} 
				 	 \leq  
				 	 \lim\limits_{\Delta\tau\downarrow 0}  \frac{\mathds{1}_{\left\{ F^{-1}(\tau ) < \varepsilon_t \leq F^{-1}(\tau + \Delta\tau) \right\} } \left( F^{-1}(\tau) - \varepsilon_t \right) }{\Delta\tau}  
				 	 \leq 0.
\end{equation}
In analogue, the left-handed limit $\lim\limits_{\Delta\tau\uparrow 0}  \frac{\text{SRAR}_{\varepsilon_t}(\tau + \Delta\tau  , F^{-1}(\tau + \Delta\tau)) - \text{SRAR}_{\varepsilon_t}(\tau, F^{-1}(\tau)) }{\Delta\tau}$ gives the same result. Therefore, we have the first-order derivative as below.
\begin{equation}
				\frac{d\, \text{SRAR}_{\varepsilon_t}(\tau, F^{-1}(\tau) )}{d \tau}
					 = \sum^{T}_{t=1}  \biggl(
					   \varepsilon_t - F^{-1}(\tau)
					   + \frac{d F^{-1}(\tau)}{d \tau} \, \left( \mathds{1}_{\left\{\varepsilon_t \leq F^{-1}(\tau)\right\} } - \tau \right) 
					\biggr).
\end{equation}
To emphasize this result, we take expectation such that
\begin{equation}
\mathbb{E}\left[ \frac{d\, \text{SRAR}_{\varepsilon_t}(\tau, F^{-1}(\tau) )}{d\, \tau} \right] = T \left( \mathbb{E}\left[ \varepsilon_t \right] - \,F^{-1}(\tau)	\right),
\label{eq:E[dSRARdtau]}				
\end{equation}
when $\mathbb{E}\left[ \varepsilon_t \right]$ exists. In practice, we are not strict with $\mathbb{E}\left[ \varepsilon_t \right]< \infty $ since the mean of an i.i.d. $\left\{ \varepsilon_t \right\}_{t=1}^T$ can be estimated empirically to replace $\mathbb{E}\left[ \varepsilon_t \right]$ in~\eqref{eq:E[dSRARdtau]} without affecting other terms.

Now we have the expectation of $ \frac{d\, \text{SRAR}_{\varepsilon_t}(\tau, F^{-1}(\tau) )}{d\, \tau}$ which can be regarded as the underlying guideline for the slope of a SRAR curve. Before interpreting this result, let us derive the second-order derivative of $\text{SRAR}_{\varepsilon_t}(\tau, F^{-1}(\tau))$ with respect to $\tau$ and make an interpretation together.

\subsubsection{The concave property}
One empirically observed property of SRAR curves is their concavity which can be explained through the second-order derivative of $\text{SRAR}_{\varepsilon_t}(\tau, F^{-1}(\tau))$ with respect to $\tau$ under assumptions (A1), (A2), (A3), (A4) and (A7). Suppose $0< \tau < \tau + \Delta\tau < 1, \Delta\tau > 0$.
\begin{equation}
\begin{aligned}
& \quad \Delta^2\,\text{SRAR}_{\varepsilon_t}(\tau, F^{-1}(\tau)) := \\
& 
\text{SRAR}_{\varepsilon_t}(\tau + \Delta\tau  , F^{-1}(\tau + \Delta\tau)) - 2\,\text{SRAR}_{\varepsilon_t}(\tau, F^{-1}(\tau)) + \text{SRAR}_{\varepsilon_t}(\tau - \Delta\tau  , F^{-1}(\tau - \Delta\tau))
					\\
					& = \sum^{T}_{t=1}\biggl(
					  \left(\varepsilon_t -  F^{-1}(\tau) \right) \left(
					   \mathds{1}_{\left\{ F^{-1}(\tau - \Delta\tau ) < \varepsilon_t \leq F^{-1}(\tau) \right\} } - \mathds{1}_{\left\{ F^{-1}(\tau ) < \varepsilon_t \leq F^{-1}(\tau + \Delta\tau) \right\} }  \right) 
					   \\
					   & + \tau \left( 2\,F^{-1}(\tau) - F^{-1}(\tau + \Delta\tau) - F^{-1}(\tau - \Delta\tau) \right) 
					   + \Delta\tau \left( F^{-1}(\tau - \Delta\tau) -  F^{-1}(\tau+ \Delta\tau)\right)
						\\					   
					   & +  \left( F^{-1}(\tau + \Delta\tau) + F^{-1}(\tau - \Delta\tau) - 2\,F^{-1}(\tau) \right) \mathds{1}_{\left\{  \varepsilon_t \leq F^{-1}(\tau - \Delta\tau )  \right\} }
					    \\
					    & + \left( F^{-1}(\tau + \Delta\tau) -  F^{-1}(\tau )\right)\mathds{1}_{\left\{ F^{-1}(\tau ) < \varepsilon_t \leq F^{-1}(\tau + \Delta\tau ) \right\} }
					\biggr)					
					.
\end{aligned}
\end{equation}
Divide the above second order central difference by $\Delta\tau^2$, and take the limit $\Delta\tau\downarrow 0$. It gives us
\begin{equation}
\label{eq:concave_formula}
\begin{aligned}
\frac{d^2\, \text{SRAR}_{\varepsilon_t}(\tau, F^{-1}(\tau) )}{d \tau^2} & = \lim\limits_{\Delta\tau\downarrow 0} \frac{\Delta^2\,\text{SRAR}_{\varepsilon_t}(\tau, F^{-1}(\tau)) }{\Delta\tau^2} 
						\\
					& = \sum^{T}_{t=1}\biggl(
					   \frac{d^2 F^{-1}(\tau)}{d \tau^2} \, \left( \mathds{1}_{\left\{\varepsilon_t \leq F^{-1}(\tau)\right\} } - \tau \right) - 2\,\frac{d F^{-1}(\tau)}{d \tau}
					\biggr),
\end{aligned}
\end{equation}
the last line of which is obtained similarly to \eqref{eq:slope_formula}. To interpret this result, we take expectation and get the following: 
\begin{equation}
\mathbb{E}\left[ \frac{d^2\, \text{SRAR}_{\varepsilon_t}(\tau, F^{-1}(\tau) )}{d\, \tau^2} \right] = - 2\,\frac{d F^{-1}(\tau)}{d \tau}\,T < 0.		
\end{equation}
where the inequality holds with probability one since $f(\epsilon)>0$ with probability one in the assumption (A1).  
Now we have the expectation of $ \frac{d^2\, \text{SRAR}_{\varepsilon_t}(\tau, F^{-1}(\tau) )}{d\, \tau^2}$ which can be regarded as the underlying guideline for the concavity of a SRAR curve. Together with the slope information, it implies that SRAR curves are always in arch shapes, going upward and then downward, with a peak point at $\mathbb{E}\left[ \varepsilon_t \right] = F^{-1}(\tau)$. We can also know the skewness of $\varepsilon_t$ from the location of the peak point: $\varepsilon_t$ is left-skewed when the SRAR curve reaches its peak in the region $\tau < 0.5$, or right-skewed when the peak in $\tau > 0.5$. If $\varepsilon_t$ is symmetrically distributed, its SRAR curve is symmetric, and vice versa.

\section{Binding functions} % warning: binding function does not exist for t(1) in mis QR(tau=0.1)
\label{sec:binding_function}
% study estmation in misspecification by concept binding function
Plotting SRAR is a way to present the goodness of fit in quantile regressions for each candidate model. Quantile regressions are the path to get residuals for SRAR calculation. As we know and provide unbiased consistent estimation for true models. To study the estimation in misspecification we adopt the concept of binding function (Dhaene, Gourieroux and Scaillet, 1998). Binding function is defined as a mapping from coefficients in the true model to pseudo-true coefficients in a misspecified model.

The estimator of a pseudo-true coefficient in quantile regression for a misspecified QCAR($p$) or QNCAR($p$) converges to a limiting value which is characterized into the binding function. It is difficult to derive the binding functions explicitly in a general case so that they are studied by means of simulations~(see Gouri\'{e}roux and Jasiak, 2017).   
Suppose a noncausal AR(1): $y_t = \pi_1 y_{t+1} + \varepsilon_t$, with $\left\{ \varepsilon_t\right\}$ i.i.d. $t(\nu)$ for $v=1, 3, 5 \,\,\text{and}\,\, 10$. It is observed that the binding function in the misspecified QCAR(1) varies with two factors: (i) the distribution of $\varepsilon_t$ and (ii) the distance function in regression which is the check function $\rho_{\tau}(\cdot)$ in quantile regression. Figure~\ref{fig:binding_function_tv_tau05}, Figure~\ref{fig:binding_function_tv_tau03} and Figure~\ref{fig:binding_function_tv_tau01} illustrate the effect of those factors. Each point is an average value of estimates based on 1000 simulations and 600 observations. Since $t(\nu)$ is symmetric, the estimation results are in the same pattern for negative true coefficient region and $\left( 1 - \tau\right)$th-quantile regression as in these three figures. Sometimes the binding function is not injective, which is evidenced in Figure~\ref{fig:binding_function_tv_tau05} and Figure~\ref{fig:binding_function_tv_tau03} for small absolute true coefficients. The non-injectivity of the binding function for Cauchy distributed innovations is also illustrated in Gouri\'{e}roux and Jasiak (2017) result that disables encompassing tests. On the other hand, we see that on Figure~\ref{fig:binding_function_tv_tau01} the injectivity of binding functions seems recovered at $\tau=10\%$. In the case of Cauchy distributed innovations, there are no binding functions from extreme quantile regressions like $0.1\,$th- or $0.9\,$th-quantile regression because the estimate is not convergent. Although a value for $\pi_1\in(0,1)$ is plotted in Figure~\ref{fig:binding_function_tv_tau01}, it is just the average of binding function estimates for $\pi_1$ for illustration. 
%\\
%
%In comparison of these three figures, we know that binding functions vary in different quantile regressions. It is not informative enough if we only consider results from one particular quantile regression, but cautious to consider binding functions in extreme quantile regressions as they may not exist. %Consequently, the results shown before lead us to consider information from all possible quantile regressions in order to stay informative and powerful.

\begin{figure}[hptb]
\centering
\includegraphics[height=7cm, width=10cm]{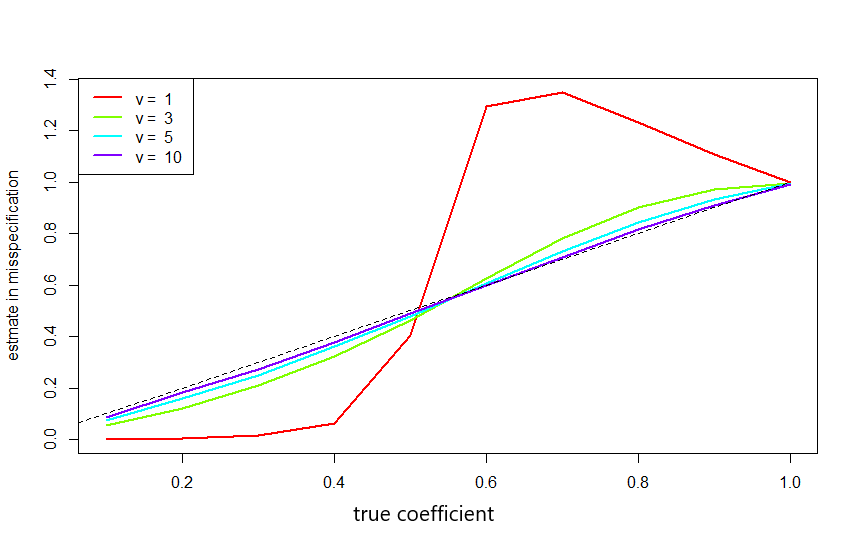} 
\caption{Binding function for a misspecified QCAR(1) in $0.5$th-quantile regression}%
\label{fig:binding_function_tv_tau05}
\end{figure}

\begin{figure}[hptb]
\centering
\includegraphics[height=7cm, width=10cm]%
{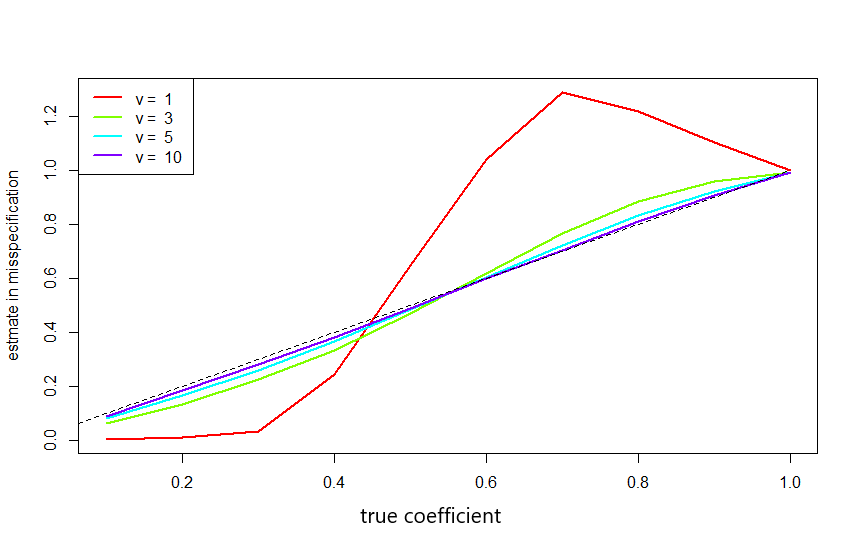}%
\caption{Binding function for a misspecified QCAR(1) in $0.3$th-quantile regression}%
\label{fig:binding_function_tv_tau03}
\end{figure}

\begin{figure}[hptb]
\centering
\includegraphics[height=7cm, width=10cm]%
{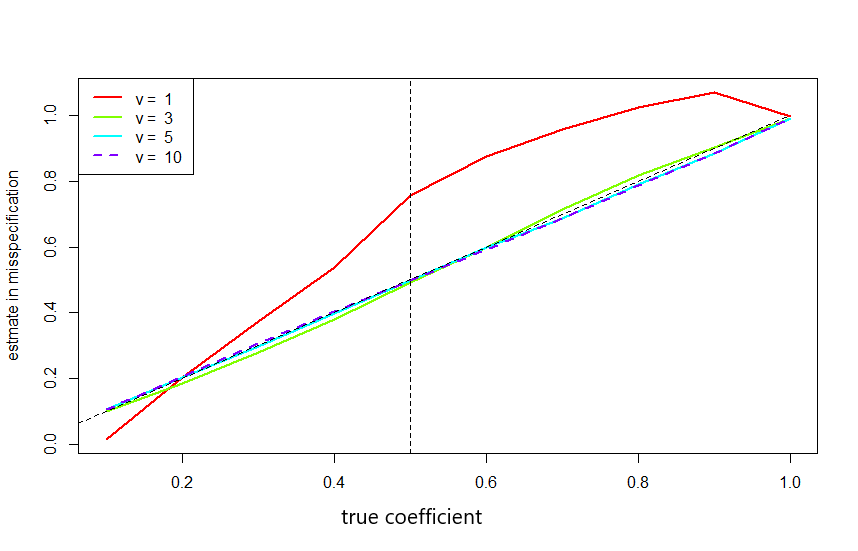}%
\caption{Binding function for a misspecified QCAR(1) in $0.1$th-quantile regression}%
\label{fig:binding_function_tv_tau01}
\end{figure}

% \subsection{MARX method v.s. the SRAR criterion}
%In the introduction we emphasized that the maximum likelihood estimation implemented to identify noncausal models depends heavily on the assumed distribution. In particular most papers rely on a symmetric distribution with heavy tails such as the Student's $t$ or the Cauchy. Such reliance of those methods may cause model selection performance vulnerable to a general distribution case in time series. Now let me check performance of model selection using the Student's $t$ likelihood implemented in the \texttt{MARX} package. 

\section{Modelling hyperinflation in Latin America}
\label{sec:empirical_analysis}

\subsection{The model specification}

The motivation of our empirical analysis comes from the rational expectation
(RE) hyperinflation model originally proposed by Cagan (1956) and investigated
by several authors (see e.g. Adam and Szafarz, 1992; Broze and Szafarz, 1985).
We follow Broze and Szafarz (1985) notations with
\begin{equation}
m_{t}^{d}=\alpha p_{t}+\beta E(p_{t+1}|I_{t})+x_{t}.\label{cagan}%
\end{equation}
In (\ref{cagan}), $m_{t}^{d}$ and $p_{t}$ respectively denote the logarithms
of money demand and price, $x_{t}$ is the disturbance term summarizing the
impact of exogenous factors. $E(p_{t+1}|I_{t})$ is the rational expectation,
when it is equal to conditional expectation, of $p_{t+1}$ at time $t$ based on the information set $I_{t}$. Assuming that the money supply $m_{t}^{s}=z_{t}$
is exogenous, the equilibrium $m_{t}^{d}=m_{t}^{s}$ provides the following
equation for prices
\begin{align*}
p_{t} &  =-\frac{\beta}{\alpha}[E(p_{t+1}|I_{t})]+\frac{z_{t}-x_{t}}{\alpha
},\\
&  =\phi\lbrack E(p_{t+1}|I_{t})]+u_{t}.
\end{align*}
Broze and Szafarz (1985) show that a forward-looking recursive solution of
this model exists when $x_{t}$ is stationary and $|\phi|<1.$ The deviation
from that solution is called the bubble $B_{t}$ with $p_{t}=\sum_{i=0}%
^{\infty}\phi^{i}E(u_{t+i}|I_{t})]+B_{t}.$ Finding conditions under which this
process has rational expectation equilibria (forward and or backward looking)
is out of the scope of our paper. We only use this framework to illustrate the
interest of economists for models with leads components. Under a perfect
foresight scheme $E(p_{t+1}|I_{t})=p_{t+1}$ we obtain the purely noncausal model
\begin{equation}
p_{t}=\phi p_{t+1}+\tilde{\varepsilon}_{t},\label{purenoncausal}%
\end{equation}
with $\tilde{\varepsilon}_{t}=u_{t}.$ In the more general setting, for
instance when $E(p_{t+1}|I_{t})=p_{t+1}+v_{t}$ with $v_{t}$ a martingale
difference, the new disturbance term is $\tilde{\varepsilon}_{t}=v_{t}+u_{t}.$
Empirically, a specification with one lead only might be too restrictive to
capture the underlying dynamics of the observed variables. We consequently
depart from the theoretical model proposed above and we consider empirical
specifications with more leads or lags. Lanne and Luoto (2013, 2017) and Hecq
et al. (2017) in the context of the new hybrid Keynesian Phillips curve assume
for instance that $\tilde{\varepsilon}_{t}$ is a MAR($r-1,s-1$) process such
as
\begin{equation}
\rho(L)\pi(L^{-1})\tilde{\varepsilon}_{t}=c+\varepsilon_{t},\label{MAR dist}%
\end{equation}
where $\varepsilon_{t}$ is $iid$ and $c$ an intercept term. Inserting
(\ref{MAR dist}) in (\ref{purenoncausal}) we observe that if
$\tilde{\varepsilon}_{t}$ is a purely noncausal model (i.e. a MAR($0,s-1)$
with $\rho(L)=1),$ we obtain a noncausal MAR($0,s)$ motion for prices
\begin{align*}
(1-\phi L^{-1})p_{t} &  =\pi(L^{-1})^{-1}(c+\varepsilon_{t}),\\
(1-\phi L^{-1})(1-\pi_{1}L^{-1}-...-\pi_{s-1}L^{-(s-1)})p_{t} &
=c+\varepsilon_{t},
\end{align*}
We would obtain a mixed causal and noncausal model if $\rho(L)\neq1.$ Our
guess is that the same specification might in some circumstances empirically
(although not mathematically as the lag polynomial does not annihilate the
lead polynomial) gives rise to a purely causal model in small samples when the
autoregressive part dominates the lead component.

\subsection{The data and unit root testing}

We consider seasonally unadjusted quarterly Consumer Price Index (CPI) series
for four Latin American countries: Brazil, Mexico, Costa Rica and Chile.
Monthly raw price series are downloaded at the OECD database for the
largest span available (in September 2018). Despite the fact that quarterly
data are directly available at OECD, we do not consider those series as they
are computed from the unweighted average over three months of the
corresponding quarters. Hence, these data are constructed using a linear
filter, leading to undesirable properties for the detection of mixed causal
and noncausal models (see Hecq, Telg and Lieb, 2017 on this specific issue).
As a consequence, we use quarterly data computed by point-in-time sampling
from monthly variables. The first observation is 1969Q1 for Mexico, 1970Q1 for
Chile, 1976Q1 for Costa Rica and 1979Q4 for Brazil. Our last observation is
2018Q2 for every series. We do not use monthly data in this paper as monthly
inflation series required a very large number of lags to capture their dynamic
feature. Moreover, the detection of seasonal unit roots in the level of
monthly price series was quite difficult.

Applying seasonal unit root tests (here HEGY tests, see Hylleberg et al.,
1990) with a constant, a linear trend and deterministic seasonal dummies, we reject
(see Table~2 in which a * denotes a rejection of the null unit root
hypothesis at a specific frequency corresponding to 5\% significance level) the null of
seasonal unit roots in each series whereas we do not reject the null of a unit
root at the zero frequency. The number of lags of the dependent variable used
to whiten for the presence of autocorrelation is chosen by AIC. From these
results we compute quarterly inflation rates for the four countries in annualized rate, i.e. $\Delta\ln P_{t}^{i}\times400.$ Next we carry out a
regression of $\Delta\ln P_{t}^{i}\times400$ on seasonal dummies to capture
the potential presence of deterministic seasonality. The null of no
deterministic seasonality is not rejected for the four series. Figure~\ref{fig:QIR_tsplot_4countries} displays quarterly inflation rates and it illustrates the huge inflation
episodes that the countries had faced. Among the four inflation rates, Brazil
and Mexico show the typical pattern closer to the intuitive notion of what a
speculative bubble is, namely a rapid increase of the series until a maximum
value is reached before the bubble bursts.%
\begin{figure}[hptb]
\centering
\includegraphics[scale=0.6]{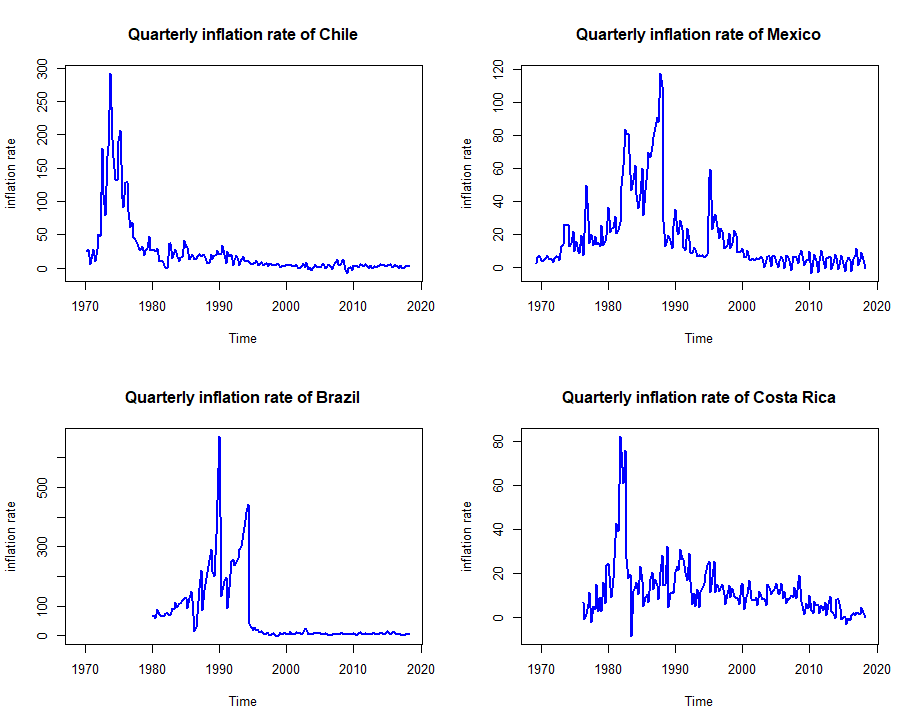} 
\caption{Quarterly inflation rate series plot for 4 Latin American countries}%
\label{fig:QIR_tsplot_4countries}
\end{figure}

\begin{table}
\centering
\label{tab:unitroottest_HEGY}
\caption{Seasonal HEGY unit root tests in the log levels of
prices}
\begin{tabular}
[c]{|lccc|c|}\hline\hline
Country & H$_{0}:\pi_{1}=0$ & H$_{0}:\pi_{2}=0$ & H$_{0}:\pi_{3}=\pi_{4}=0$ &
Sample\\\hline
$\ln P_{t}^{Bra}$ & $-1.39$ & $-5.75\ast$ & $48.28\ast$ & $1979Q4-2018Q2$\\
$\ln P_{t}^{Chi}$ & $-2.98$ & $-6.32\ast$ & $20.13\ast$ & $1970Q1-2018Q2$\\
$\ln P_{t}^{Costa}$ & $-1.80$ & $-4.23\ast$ & $7.81\ast$ & $1976Q1-2018Q2$\\
$\ln P_{t}^{Mex}$ & $-0.88$ & $-11.92\ast$ & $60.10\ast$ & $1969Q1-2018Q2$%
\\\hline
\end{tabular}
\end{table}

\subsection{Empirical findings and identification of noncausal models}

Table~3 reports for each quarterly inflation rates the autoregressive model
obtained using the Hannan-Quinn information criterion. Given our results on
the binding function (see also Gouri\'{e}roux and Jasiak, 2017) it is safer to
determine the pseudo-true autoregressive lag length using such an OLS approach
than using quantile regressions or using maximum likelihood method. Indeed
there is the risk that a regression in direct time from a noncausal DGP
provides an underestimation of the lag order for some distributions (e.g. the
Cauchy) and some values of the parameters. 

Estimating autoregressive univariate models gives the lag length range from $p=1$
for Brazil to $p=7$ for the Chilean inflation rate. The $p-values$ of the
Breush-Pagan LM test (see column labeled $LM[1-2])$ for the null of
no-autocorrelation after having included those lags show that we do not reject
the null in every four cases. On the other hand, we reject\ the null of
normality (Jarque-Bera test) in the disturbances of each series. We should
consequently be able to identify causal from noncausal models. From columns
$skew.$ and $kurt.$ it emerges that the residuals are skewed to the left for
Brazil and Mexico and skewed to the right for Chile and Costa Rica. Heavy
tails are present in each series. At a 5\% significance level we reject the
null of no ARCH (see column $ARCH[1-2])$ for Brazil and Mexico. Gouri\'{e}roux
and Zakoian (2017) have derived the closed form conditional moments of a
misspecified causal model obtained from a purely noncausal process with alpha
stable disturbances. They show that the conditional mean (in direct time) is a
random walk with a time varying conditional variance in the Cauchy case. This
result would maybe favour the presence of a purely noncausal specification for
Brazil and Mexico as the null of no ARCH is rejected. But this assertion must
be carefully evaluated and tested, for instance using our comparison of
quantile autoregressions in direct and reverse time. The results by the Q(N)CAR are reported in Table~4, and the RQ(N)CAR produces the same results. Each cell of Table~4 provides the selection frequency of MAR($0,p$) or
MAR$(p,0$) identified by the SRAR at quantiles 0.1, 0.3, 0.5, 0.7, 0.9 as well as the aggregated SRAR. Figure~\ref{fig:SRAR_QIR_4countries} displays the SRAR curves from 0.05th-quantile to 0.95th-quantile by the Q(N)CAR for the four economies respectively, similarly to the ones by the RQ(N)CAR with restriction on non-negative regressors. As observed, the crossing feature appears in the SRAR plots. Especially in the SRAR plot for Brazil, it is hard to trust a model from evidence at single quantiles. However, the aggregate SRAR criterion comes to help for this situation from an overall perspective. We conclude that Brazil, Mexico and Costa Rica are better characterized as being purely noncausal while Chile being purely noncausal according to the aggregate SRAR criterion.

%%% SRAR plot in empirical study by QAR
\begin{figure}[hptb]
\centering
\includegraphics[scale=0.6]{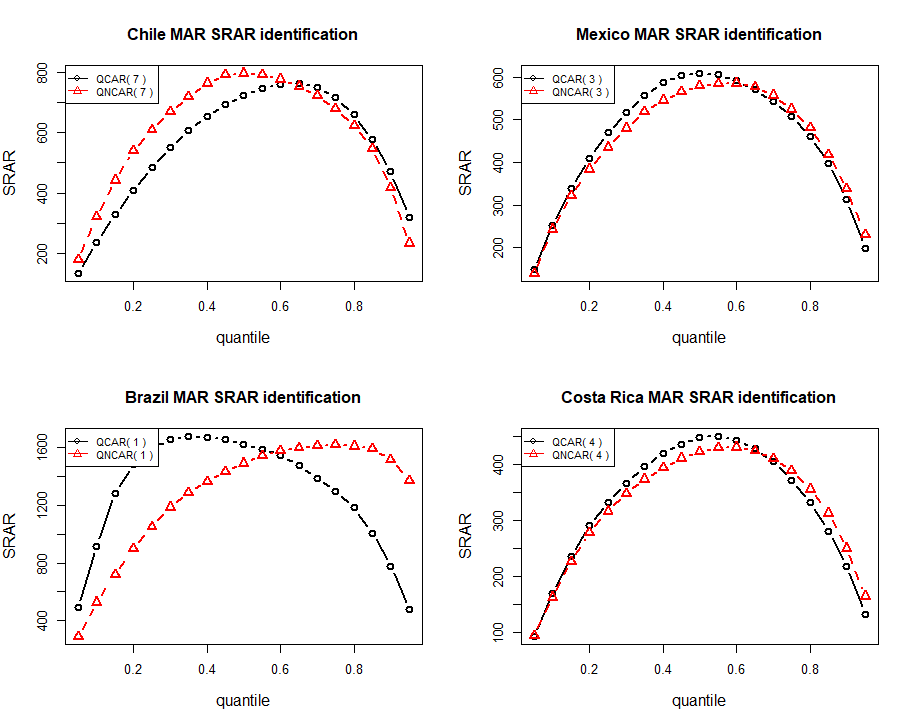} 
\caption{SRAR plots of the inflation rates of four Latin American countries respectively}%
\label{fig:SRAR_QIR_4countries}
\end{figure}

%%%% SRAR plot in empirical study by RQAR
%\begin{figure}[hptb]
%\centering
%\includegraphics[scale=0.8]{graphics/mar_srar_identify_4countries_omitneg.png} 
%\caption{SRAR plots of the inflation rates of four Latin American countries respectively by restriction}%
%\label{fig:SRAR_QIR_4countries_RQAR}
%\end{figure}
%
\begin{table}
\label{tab:QIR_descr_stat}
\centering
\caption{Descriptive statistics for quarterly inflation rates}
\begin{tabular}
[c]{|lcccccc|}\hline\hline
Country & $HQ$ & $BJ$ & $skew.$ & $kurt.$ & $LM[1-2]$ & $ARCH[1-2]$\\\hline
$\Delta\ln P_{t}^{Bra}$ & $1$ & $<0.001$ & $-2.54$ & $56.96$ & $0.19$ &
$<0.001$\\
$\Delta\ln P_{t}^{Chi}$ & $7$ & $<0.001$ & $2.84$ & $22.45$ & $0.09$ &
$0.09$\\
$\Delta\ln P_{t}^{Costa}$ & $4$ & $<0.001$ & $1.01$ & $8.73$ & $0.47$ &
$0.30$\\
$\Delta\ln P_{t}^{Mex}$ & $3$ & $<0.001$ & $-0.40$ & $13.81$ & $0.20$ &
$<0.001$\\\hline
\end{tabular}
\end{table}

\begin{table}
\label{tab:MAR_identify_by_srar}
\caption{SRAR identification results}\vspace{5pt}
{\small
\begin{tabular}
[c]{|lccccc|c|}\hline\hline
Country & $SRAR_{\tau=0.1}$ & $SRAR_{\tau=0.3}$ & $SRAR_{\tau=0.5}$ &
$SRAR_{\tau=0.7}$ & $SRAR_{\tau=0.9}$ & $SRAR_{total}$\\\hline
$\Delta\ln P_{t}^{Bra}$ & $MAR(0,1)$ & $MAR(0,1)$ & $MAR(0,1)$ & $MAR(1,0)$ & $MAR(1,0)$
& $MAR(0,1)$\\
$\Delta\ln P_{t}^{Chi}$ & $MAR(7,0)$ & $MAR(7,0)$ & $MAR(7,0)$ & $MAR(0,7)$ & $MAR(0,7)$
& $MAR(7,0)$\\
$\Delta\ln P_{t}^{Costa}$ & $MAR(0,4)$ & $MAR(0,4)$ & $MAR(0,4)$ & $MAR(4,0)$ &
$MAR(4,0)$ & $MAR(0,4)$\\
$\Delta\ln P_{t}^{Mex}$ & $MAR(0,3)$ & $MAR(0,3)$ & $MAR(0,3)$ & $MAR(3,0)$ & $MAR(3,0)$
& $MAR(0,3)$\\\hline
\end{tabular}
}
\end{table}

\section{Conclusions}
\label{sec:conclusion}
This paper introduces a new way to detect noncausal from causal models by comparing residuals from quantile autoregressions developed by Koenker and Xiao (2006) and from the time-reverse specifications. To adapt to heavy tailed distributions, we generalize the quantile autoregression theory for regularly varying distributions. This also confirms the validity of quantile autoregressions in analysing heavy tailed time series, such as explosive or bubble-type dynamics. It is natural to consider SRAR as a model selection criterion in the quantile regression framework. However due to the crossing feature of SRAR plots as presented in this paper, we propose to use the aggregate SRAR criterion for model selection. The robustness in its performance has been seen from all the results in this paper. %On the other hand, the study on misspecification has been studied from the perspective of binding functions by simulations as it is formidable to obtain analytical solution to binding functions in semi-parametric methods. We know that the encompassing tests are hard to be implemented because of no explicit formula for binding functions generally. Quantile regressions are more informative in exploring data than maximum likelihood method (Gouri\'{e}roux, C. and J. Jasiak, 2017). Having the same number of parameters in both candidate models, we propose to use the aggregate SRAR as model selection criterion in model identifying. 
In the empirical study on the inflation rates of four Latin American countries, we found that the purely noncausal specification is favoured in three cases. 
%The mixed model is currently under investigation.

Finally a possible extension of our approach is the identification of mixed models in addition to purely causal and noncausal specifications. Also, a formal testing on SRAR differences would require the application of a bootstrap approach that is beyond the scope of our paper.

\appendix                                                                                                               
\section*{Appendix}

\subsection*{Alternative way to simulate MAR models}
Suppose that the DGP is a MAR($r,s$) as in~\eqref{MAR}. First, we rewrite~\eqref{MAR} into a matrix representation as follows:
\begin{equation}
\begin{aligned}
\mathbf{M}\boldsymbol{y} & = \boldsymbol{\varepsilon},\\
\mathbf{M} & := \begin{bmatrix} 
\pi(L)\phi(L^{-1}) & 0 &\ldots & 0 \\
0 & \pi(L)\phi(L^{-1}) &\ldots & 0 \\
	& \ldots	&	&	\\
0	& 0	& \ldots	& \pi(L)\phi(L^{-1})	
\end{bmatrix},\\
\boldsymbol{y} & := \begin{bmatrix} 
y_1 & y_2 & \ldots	& y_T
\end{bmatrix}' ,\\
\boldsymbol{\varepsilon} & := \begin{bmatrix} 
\varepsilon_1 & \varepsilon_2 & \ldots	& \varepsilon_T
\end{bmatrix}' ,
\end{aligned}
\label{eq:MAR_matrix_rep}
\end{equation}
where $\mathbf{M}$ is $T\times T$ matrix and $T$ is the sample size. The equivalence to ~\eqref{MAR} holds by assuming $ y_{1-r}, y_{2-r}, \ldots, y_{0}$ and $y_{T+1}, y_{T+2}, \ldots, y_{T+s}$ are all zeros. This assumption effect can be neglected by deleting enough observations from the beginning and the end of a simulated sample, for instance, $\left\{ y_t \right\}_{t = 201}^{T-200}$ kept for analysis from a first simulated $\left\{ y_t \right\}_{t = 1}^T$. Next, $\mathbf{M}$ can be decomposed into a product of two diagonal matrices, denoted as $\mathbf{L}$ and $\mathbf{U}$, of main diagonal entries being $\pi(L)$ and $\phi(L^{-1})$ respectively as follows. 
\begin{equation}
\begin{aligned}
\mathbf{L} & = \begin{bmatrix} 
1 		& 0 		& 0 		&  0	 		&\ldots &	 	& 0 \\
-\pi_1 	& 1			& 0 		&  0			&\ldots &	 	& 0 \\
-\pi_2 	&-\pi_1 	& 1 		&  0 		&\ldots &	 	& 0 \\
		& \ldots	&			&		 	&\ldots & 		&   \\
0		& 		& \ldots 	& -\pi_r 	&\ldots & -\pi_1	& 1	
\end{bmatrix},\\
\mathbf{U} & = \begin{bmatrix} 
1 	& -\psi_1 	& \ldots 	& -\psi_s 	& 0	 		&  	 & \ldots	& 0\\
0 	& 1			& -\psi_1 	& \ldots 	& -\psi_s 	& 0 &	\ldots 	& 0 \\
\ldots 	& 			& \ldots  & 	&  	&  &	\ldots 	&  \\
0 	& 			&  \ldots & 	&  	& 0 &	\ldots 	& 1
\end{bmatrix}
\end{aligned}
\label{eq:decompose_A}
\end{equation}
Substitute \eqref{eq:decompose_A} into \eqref{eq:MAR_matrix_rep}. We get
\begin{equation*}
\mathbf{L}\mathbf{U}\boldsymbol{y} = \boldsymbol{\varepsilon},
\end{equation*}
such as
\begin{equation}
\boldsymbol{y} = \mathbf{U}^{-1}\mathbf{L}^{-1}\,\boldsymbol{\varepsilon}\,.
\label{eq:dgp_formula_mar}
\end{equation}
Given $\boldsymbol{\varepsilon}$, $\boldsymbol{y}$ can be obtained directly since $\mathbf{L}$ and $\mathbf{U}$ are positive definite triangular matrices. This $MAR(r,s)$ simulating method can easily be generalized, for instance, for an $MAR(r,s)$ involving some exogenous independent variables presented by Hecq, Issler and Telg~(2017). In practice this vector-wise simulation method is slower than the element-wise method because of the matrix creation and storage in simulation.

\subsection*{Proof of Theorem~\ref{thm:asym_regvary_NCAR}}
\begin{proof}
$ $\newline
First, we rewrite SRAR($\tau, \hat{\boldsymbol{\theta}}(\tau )$) as follows:
\begin{equation}
\begin{aligned}
\text{SRAR}(\tau,\hat{\boldsymbol{\theta}}(\tau ))	& = \sum\limits_{t=1}^{T}\rho_{\tau}(y_{t}-\boldsymbol{x}_{t}^{\prime}\hat{\boldsymbol{\theta}}(\tau ))\\
				& = \sum\limits_{t=1}^{T}\rho_{\tau} (y_{t}- \boldsymbol{x}_{t}^{\prime}\boldsymbol{\phi}_{\tau} + \boldsymbol{x}_{t}^{\prime}\boldsymbol{\phi}_{\tau} - \boldsymbol{x}_{t}^{\prime}\hat{\boldsymbol{\theta}}(\tau ) )
				\\
						& =  \sum\limits_{t=1}^{T}\rho_{\tau}\left( u_{t\tau} - \frac{1}{a_T\sqrt{T}}\boldsymbol{\nu}'\boldsymbol{x}_t\right),
\end{aligned}
 \label{eq:SRAR_phi_theta}%
\end{equation}
where $ \boldsymbol{x}_{t}^{\prime} :=\left[  a_T,y_{t+1},\ldots,y_{t+p}\right] $, $u_{t\tau}:= y_{t}-\boldsymbol{x}_{t}^{\prime}\boldsymbol{\phi}_{\tau} = \varepsilon_t - F^{-1}(\tau)$, $\boldsymbol{\nu}: = a_T\sqrt{T}\left( \hat{\boldsymbol{\theta}}(\tau ) - \boldsymbol{\phi}_{\tau}\right)$. We know from Davis and Resnick (1985) and Knight (1989, 1991) that
\begin{equation}
\begin{aligned}
\frac{1}{a_T}\left( \sum\limits_{t=1}^{\floor{T\cdot s}}\left( \varepsilon_t - b_T\right)	\right)		& \overset{d}{\sim}	\mathrm{S}_{\alpha}(s),
					\\
\frac{1}{a_T\,\sqrt{T}}\sum\limits_{t=1}^{T}\left( y_{t} - \floor{T\cdot s}\sum\limits^{\infty}_{j=0} c_j b_T\right) 	& \overset{d}{\sim}\sum\limits^{\infty}_{j=0} c_j\, \int^1_0 \mathrm{S}_{\alpha}(s)\, ds,
						\\
\frac{1}{a_T^2\,T}\sum\limits_{t=1}^{T}\left( y_t\cdot y_{t+h} - \floor{T\cdot s}\sum\limits^{\infty}_{j=0} c_j\, c_{j+h} b_T^2\right)	& \overset{d}{\sim}\sum\limits^{\infty}_{j=0}  c_j\, c_{j+h}\, \int^1_0 \mathrm{S}_{\alpha}^2(s)\, ds,		
\end{aligned}
\label{eq:dist_limit_MA}
\end{equation}
where $ t= \floor{T\cdot s}$, and $\left\{\mathrm{S}_{\alpha}(s)\right\}$ is a process of stable distributions with index $\alpha$. Without loss of generality, we assume $b_T=0$ for the proof below. In use of the limiting behaviour information presented in~\eqref{eq:dist_limit_MA}, we get that
\begin{equation}
\begin{aligned}
\frac{1}{a_t^2 T}\sum^{T}_{t=1}\boldsymbol{x}_t\boldsymbol{x}^{\prime}_{t} \overset{d}{\sim}  & = \begin{bmatrix} 
1 		& \boldsymbol{0} 	\\
\boldsymbol{0}		& \Omega_{\boldsymbol{S}}\, \Omega_1
\end{bmatrix}_{(p+1)\times (p+1)} 
\end{aligned}
\end{equation}
where 
\begin{equation*}
\begin{aligned}
%\Omega & = \begin{bmatrix} 
%1 		& \boldsymbol{0} 	\\
%\boldsymbol{0}		& \Omega_1
%\end{bmatrix}_{(p+1)\times (p+1)},\\
\Omega_1 & := \bigl[
\omega_{ik}
\bigr]_{p\times p}, \\
\omega_{ik} &:= \sum\limits^{\infty}_{j=0}  c_j\, c_{j+|k-i|},  \\
\Omega_{\boldsymbol{S} } &:= \text{diag}\left( \int_0^1 \mathrm{S}_{\alpha}^2(s)\,ds, \int_0^1 \mathrm{S}_{\alpha}^2(s + \frac{1}{T})\,ds,\ldots, \int_0^1 \mathrm{S}_{\alpha}^2(s + \frac{p-1}{T})\,ds\right),
\end{aligned}
\end{equation*}
with $\omega_{ik}$ being the element at $\Omega_1$'s $i$-th row and $k$-th column, and $\Omega_{\boldsymbol{S} }$ being a $p\times p$ diagonal matrix with the $j-$th diagonal entry being $ \int_0^1 \mathrm{S}_{\alpha}^2(s + \frac{j-1}{p})\,ds $, $j\in\left\{1,2,\ldots,p \right\}$.  $\Omega_1$ is positive definite symmetric.
% and therefore $\Omega$ is positive definite symmetric as well.
Note that $\hat{\boldsymbol{\theta}}(\tau)=\argmin\limits_{\boldsymbol{\theta}\in\mathbb{R}^{p+1}}\;\text{SRAR}(\tau,\boldsymbol{\theta})$ which also minimizes
\begin{equation}
Z_{T}(\boldsymbol{\nu}):= \sum\limits_{t=1}^{T} \left[ \rho_{\tau}\left( u_{t\tau} - \frac{1}{a_T\sqrt{T}}\boldsymbol{\nu}'\boldsymbol{x}_t\right) - \rho_{\tau}\left( u_{t\tau} \right) \right].
\label{eq:equiv_mini}
\end{equation}
$Z_{T}(\boldsymbol{\nu})$ is a convex random function. Knight (1989) showed that if $Z_{T}(\boldsymbol{\nu})$ converges in distribution to $Z(\boldsymbol{\nu})$ and $Z(\boldsymbol{\nu})$ has unique minimum, then the convexity of $Z_{T}(\boldsymbol{\nu})$ ensures 
%$$
%\argmin\limits_{\boldsymbol{\nu}\in \mathbb{R}^{p+1}} \overset{\sim}{d} Z_{T}(\boldsymbol{\nu}) \argmin\limits_{\boldsymbol{\nu}\in \mathbb{R}^{p+1}} Z(\boldsymbol{\nu})
%$$
$\hat{\boldsymbol{\nu}} = \argmin\limits_{\boldsymbol{\nu}\in \mathbb{R}^{p+1}} Z_{T}(\boldsymbol{\nu})$ converging in distribution to $ \argmin\limits_{\boldsymbol{\nu}\in \mathbb{R}^{p+1}} Z(\boldsymbol{\nu})$. \\
By using the following check function identity:
\begin{equation}
\begin{aligned}
\rho_{\tau}(v_1 - v_2) - \rho_{\tau}(v_1)		& = -v_2 \xi_{\tau}(v_1) + (v_1 - v_2)\left( I(0> v_1 > v_2) - I(0 < v_1 < v_2)\right) 
\\
		& = -v_2 \xi_{\tau}(v_1) + \int^{v_2}_0 \left( I( v_1 \leq s ) - I( v_1 < 0)\right) ds,
\end{aligned}
\label{eq:check_func_identity}
\end{equation}
where $\xi_{\tau}(v) := \tau - I(v<0)$, we can rewrite $Z_{T}(\boldsymbol{\nu})$ into
\begin{equation}
\begin{aligned}
Z_{T}(\boldsymbol{\nu})	& = -\sum^{T}_{t=1} \frac{1}{a_T\sqrt{T}}\boldsymbol{\nu}'\boldsymbol{x}_t\, \xi_{\tau}(u_{t\tau}) + \sum^{T}_{t=1} \int^{\frac{1}{a_T\sqrt{T}}\boldsymbol{\nu}'\boldsymbol{x}_t}_0 \left( I( u_{t\tau} \leq s ) - I( u_{t\tau} < 0)\right) ds
						\\
						& = Z_{T}^{(1)}(\boldsymbol{\nu}) + Z_{T}^{(2)}(\boldsymbol{\nu}) ,			
\end{aligned}
\end{equation}
where $Z_{T}^{(2)}(\boldsymbol{\nu}):=  \sum^{T}_{t=1} \int^{\frac{1}{a_T\sqrt{T}}\boldsymbol{\nu}'\boldsymbol{x}_t}_0 \left( I( u_{t\tau} \leq s ) - I( u_{t\tau} < 0)\right) ds $ and $Z_{T}^{(1)}(\boldsymbol{\nu}):= -\sum^{T}_{t=1} \frac{1}{a_T\sqrt{T}}\boldsymbol{\nu}'\boldsymbol{x}_t\, \xi_{\tau}(u_{t\tau}) $. Further denote
$\,
\eta_t(\boldsymbol{\nu}) := \int^{\frac{1}{a_T\sqrt{T}}\boldsymbol{\nu}'\boldsymbol{x}_t}_0 \left( I( u_{t\tau} \leq s ) - I( u_{t\tau} < 0)\right) ds
$,
$\,
\bar{\eta}_t(\boldsymbol{\nu}) := E\left[ \eta_t(\boldsymbol{\nu})	| \boldsymbol{x}_t \right]
$ and $\,
\overline{Z}_{T}^{(2)}(\boldsymbol{\nu}) := \sum^{T}_{t=1}\bar{\eta}_t(\boldsymbol{\nu}).
$
% Then $\left\{ \eta_t(\boldsymbol{\nu}) - \bar{\eta}_t(\boldsymbol{\nu})\right\}$ is a martingale difference sequence.
By Assumption (A5) and small enough $\frac{1}{a_T\sqrt{T}}\boldsymbol{\nu}'\boldsymbol{x}_t$, we further rewrite $\overline{Z}_{T}^{(2)}(\boldsymbol{\nu}) $ as follows:
\begin{equation}
\begin{aligned}
\overline{Z}_{T}^{(2)}(\boldsymbol{\nu}) & = \sum^{T}_{t=1}E\left[ \int^{\frac{1}{a_T\sqrt{T}}\boldsymbol{\nu}'\boldsymbol{x}_t}_0 \left( I( u_{t\tau} \leq s ) - I( u_{t\tau} < 0)\right) ds	\middle| \boldsymbol{x}_t \right]		\\
					& = \sum^{T}_{t=1}  \int^{\frac{1}{a_T\sqrt{T}}\boldsymbol{\nu}'\boldsymbol{x}_t}_0 \left[ \int^{s+F^{-1}(\tau)}_{F^{-1}(\tau)} f(r) dr \right] ds		\\
					& = \sum^{T}_{t=1}  \int^{\frac{1}{a_T\sqrt{T}}\boldsymbol{\nu}'\boldsymbol{x}_t}_0	\frac{F\left( s+F^{-1}(\tau)\right) - F\left( F^{-1}(\tau)\right)}{s} s\, ds		\\
					& = \sum^{T}_{t=1}  \int^{\frac{1}{a_T\sqrt{T}}\boldsymbol{\nu}'\boldsymbol{x}_t}_0	f\left( F^{-1}(\tau)\right) s\, ds			\\
					& = \frac{1}{2 a_T^2\,T}\sum^{T}_{t=1} f\left( F^{-1}(\tau)\right)\boldsymbol{\nu}'\boldsymbol{x}_t\boldsymbol{x}^{\prime}_{t}\boldsymbol{\nu} + o_p(1)	\\
					& = \frac{1}{2 a_T^2\,T}\,f\left( F^{-1}(\tau)\right)\,\boldsymbol{\nu}' \left(\sum^{T}_{t=1}\boldsymbol{x}_t\boldsymbol{x}^{\prime}_{t} \right)\boldsymbol{\nu} + o_p(1)
\end{aligned}
\end{equation}
Using the limiting behaviour information presented in~\eqref{eq:dist_limit_MA}, we get the limiting distribution for $\overline{Z}_{T}^{(2)}(\boldsymbol{\nu})$ so as for $ Z_{T}^{(2)}(\boldsymbol{\nu})  $ as follows:
\begin{equation}
\begin{aligned}
Z_{T}^{(2)}(\boldsymbol{\nu})  \overset{d}{\sim}	\frac{1}{2} f\left( F^{-1}(\tau)\right)\,\boldsymbol{\nu}'   \begin{bmatrix} 
1 		& \boldsymbol{0} 	\\
\boldsymbol{0}		& \Omega_{\boldsymbol{S}}\, \Omega_1
\end{bmatrix} \,\boldsymbol{\nu},
\label{eq:asym_WT}
\end{aligned}
\end{equation}
by the fact that $ Z_{T}^{(2)}(\boldsymbol{\nu}) - \overline{Z}_{T}^{(2)}(\boldsymbol{\nu})  \overset{p}{\sim}\,0\;$ which can be proved by following the arguments of Knight (1989).
\\
The limiting distribution of $ Z_{T}^{(1)}(\boldsymbol{\nu})$ can also be deduced in using ~\eqref{eq:dist_limit_MA} as follows.
\begin{equation}
\begin{aligned}
-\sum^{T}_{t=1} & \frac{1}{a_T\sqrt{T}}\boldsymbol{\nu}'\boldsymbol{x}_t\, \xi_{\tau}(u_{t\tau}) \overset{d}{\sim}  \\
			& \boldsymbol{\nu}' \left[ \sigma_{\xi}  W(1), \sum\limits^{\infty}_{j=0} c_j\,\sigma_{\xi} \int^1_0 \mathcal{S}_{\alpha}(s)\,dW(s)\,, \ldots, \sum\limits^{\infty}_{j=0} c_j\,\sigma_{\xi} \int^1_0 \mathcal{S}_{\alpha}(s + \frac{p-1}{T})\,dW(s) \right]_{(p+1)\times 1}, 
\end{aligned}
\end{equation}
where $\left[ \ldots\right]_{(p+1)\times 1}$ is a column vector of $\left(p+1\right)$ elements, $\int dW(s)$ is a stochastic integral with Brownian motion $\left\{W(s) \right\}$ independent of $\left\{\mathrm{S}_{\alpha}(s)\right\}$ (see Knight (1991)), and $\sigma_{\xi}$ is the standard deviation of $  \xi_{\tau}(u_{t\tau})$ which equals $\sqrt{\tau(1-\tau)}$. Therefore by Davis and Resnick (1985) and Knight (1989, 1991), 
\begin{equation}
Z_{T}^{(1)}(\boldsymbol{\nu}) 
\overset{d}{\sim} 
\boldsymbol{\nu}' \sqrt{\tau(1-\tau)}\left[ W(1), \sum\limits^{\infty}_{j=0} c_j\, \int^1_0 \mathcal{S}_{\alpha}(s)\,dW(s)\,, \ldots, \sum\limits^{\infty}_{j=0} c_j\, \int^1_0 \mathcal{S}_{\alpha}(s + \frac{p-1}{T})\,dW(s) \right]_{(p+1)\times 1}.
\end{equation}
Thus,
\begin{equation}
\begin{aligned}
Z_{T}(\boldsymbol{\nu})	& \overset{d}{\sim} 
								Z(\boldsymbol{\nu}):= \\														&  \boldsymbol{\nu}' \sqrt{\tau(1-\tau)} \left[ W(1), \sum\limits^{\infty}_{j=0} c_j\, \int^1_0 \mathcal{S}_{\alpha}(s)\,dW(s)\,, \ldots, \sum\limits^{\infty}_{j=0} c_j\, \int^1_0 \mathcal{S}_{\alpha}(s+ \frac{p-1}{T})\,dW(s) \right]_{(p+1)\times 1}
				\\
				& + \frac{1}{2} f\left( F^{-1}(\tau)\right)\,\boldsymbol{\nu}'   \begin{bmatrix} 
1 		& \boldsymbol{0} 	\\
\boldsymbol{0}		& \Omega_{\boldsymbol{S}}\, \Omega_1
\end{bmatrix} \,\boldsymbol{\nu}.
\end{aligned}
\end{equation}
and so
\begin{equation*}
\begin{aligned}
	\frac{f\left( F^{-1}(\tau)\right) \cdot a_T\sqrt{T}}{\sqrt{\tau(1-\tau)}\,}  & \left(\hat{\boldsymbol{\theta}}(\tau ) - \boldsymbol{\phi}_{\tau} \right) \overset{d}{\sim}  \qquad\qquad	\\
 \begin{bmatrix} 
1 		& \boldsymbol{0} 	\\
\boldsymbol{0}		& \Omega_1^{-1}\Omega_{\boldsymbol{S}}^{-1}
\end{bmatrix} & \left[ W(1), \sum\limits^{\infty}_{j=0} c_j\, \int^1_0 \mathcal{S}_{\alpha}(s)\,dW(s)\,, \ldots, \sum\limits^{\infty}_{j=0} c_j\, \int^1_0 \mathcal{S}_{\alpha}(s+\frac{p-1}{T})\,dW(s) \right]_{(p+1)\times 1}.
\end{aligned}
\end{equation*}
follows by setting the derivative of $Z(\boldsymbol{\nu})$ to $0$ and solving for $\boldsymbol{\nu}$.
\end{proof}

\end{document}